\documentclass[letterpaper, 10 pt, conference]{ieeeconf}  

\IEEEoverridecommandlockouts                              
\usepackage{graphics} 
\usepackage{epsfig} 
\usepackage{times} 
\usepackage{amsmath} 
\usepackage{amssymb}  
\usepackage{microtype}
\usepackage{graphicx}
\usepackage{subfig}
\usepackage{booktabs} 
\usepackage{amsfonts}
\usepackage{amsmath} 
\usepackage{mathrsfs}
\usepackage{doi}
\usepackage{bm}
\usepackage{float}

\usepackage{algorithm}
\usepackage{algorithmic}
\usepackage{xcolor}
\usepackage{cite}

\newtheorem{theorem}{Theorem}
\newtheorem{lemma}{Lemma}

\newtheorem{assumption}{Assumption}

\title{\LARGE \bf
Model-Reference Reinforcement Learning Control of Autonomous Surface Vehicles with Uncertainties}

\author{Qingrui Zhang$^{1,2}$,  Wei Pan$^{2}$, and Vasso Reppa$^{1}$
\thanks{$^{1}$Department of Maritime and Transport Technology, Delft University of Technology, Delft, the Netherlands
        {\tt\small Qingrui.Zhang@tudelft.nl; V.Reppa@tudelft.nl}}%
\thanks{$^{2}$Department of Cognitive Robotics, Delft University of Technology, Delft, the Netherlands
        {\tt\small Wei.Pan@tudelft.nl}}%
}

\begin{document}

\maketitle
\thispagestyle{empty}
\pagestyle{empty}

\begin{abstract}
This paper presents a novel model-reference reinforcement learning control method for uncertain autonomous surface vehicles. The proposed control combines a conventional control method with deep reinforcement learning.  With the conventional control, we can ensure the learning-based control law provides closed-loop stability for the overall system, and potentially increase the sample efficiency of the deep reinforcement learning. With the reinforcement learning, we can directly learn a control law to compensate for modeling uncertainties. In the proposed control, a nominal system is employed for the design of a baseline control law using a conventional control approach. The nominal system also defines the desired performance for uncertain autonomous vehicles to follow.  In comparison with traditional deep reinforcement learning methods, our proposed learning-based control can provide stability guarantees and better sample efficiency. We demonstrate the performance of the new algorithm via extensive simulation results.
\end{abstract}


\section{INTRODUCTION}
Autonomous surface vehicles (ASVs) have been attracting more and more  attention, due to their advantages in many applications, such as environmental monitoring \cite{Jones2019STTE}, resource exploration \cite{Majohr2006}, shipping \cite{Levander2017IEEE}, and many more. Successful launch of ASVs in real life requires accurate tracking control along a desired trajectory \cite{Do2004Automatica,Do2006Auto, Sonnenburg2013JFR}. However, accurate tracking control for ASVs is challenging, as ASVs are subject to uncertain nonlinear hydrodynamics and unknown environmental disturbances \cite{Fossen2011Handbook}. Hence, tracking control of highly uncertain ASVs has received extensive research attention \cite{Soltan2009ACC, Yu2012IET, Wang2016TC, Wang2017OE, Woo2019OE}.

Control algorithms for uncertain systems including ASVs mainly lie in four categories: 1) robust control which is the ``worst-case''  design for bounded uncertainties and disturbances \cite{Yu2012IET}; 2) adaptive control which adapts to system uncertainties with parameter estimations \cite{Do2004Automatica, Do2006Auto}; 3) disturbance observer-based control which compensates uncertainties and disturbances in terms of the observation technique \cite{Wang2017OE, Zhang2018TIE}; and 4) reinforcement learning (RL) which learns a control law from data samples \cite{Woo2019OE, Shi2019TNNLS}. The first three algorithms follow a model-based control approach, while the last one is data driven. Model-based control can ensure closed-loop stability, but a system model is indispensable. Uncertainties and disturbances of a system should also satisfy different conditions for different model-based methods. In robust control, uncertainties and disturbances are assumed to be bounded with known boundaries \cite{Shen1995TAC}. As a consequence, robust control will lead to conservative high-gain control laws which usually limits the control performance (i.e., overshoot, settling time, and stability margins) \cite{Liu2008CDC}. Adaptive control can handle varying uncertainties with unknown boundaries, but system uncertainties are assumed to be linearly parameterized with known structure and unknown constant parameters \cite{Haddad2002IJACSP,Zhang2018AST}.  A valid adaptive control design also requires a system to be persistently excited, resulting in the unpleasant high-frequency oscillation behaviours in control actions \cite{Ioannou1996}. On the other hand, disturbance observer-based control  can adapt to both uncertainties and disturbances with unknown structures and without assuming systems to be persistently excited \cite{Zhang2018TIE, Zhu2018IJRNC}.  However, we need the frequency information of uncertainty and disturbance signals when choosing proper gains for the disturbance observer-based control, otherwise it is highly possible to end up with a high-gain control law \cite{Zhu2018IJRNC}. In addition, the disturbance observer-based control can only address matched uncertainties and disturbances, which act on systems through the control channel \cite{Zhang2018AST, Mondal2013ISA}. In general, comprehensive modeling and analysis of systems are essential for all model-based methods. 

In comparison with model-based methods, RL is capable of learning a control law from data samples using much less model information \cite{Sutton2018MIT}. Hence, it is more promising in controlling systems subject to massive uncertainties and disturbances as ASVs  \cite{Woo2019OE, Shi2019TNNLS, Meyer2019arXiv_ASV, Zhou2019Access}, given the sufficiency and good quality of collected data. Nevertheless, it is challenging for model-free RL to ensure closed-loop stability, though some research attempts have been made \cite{han2019hinf}.  It implies that the learned control law must be re-trained, once some changes happen to the environment or the reference trajectory (i.e. in \cite{Shi2019TNNLS}, the authors conducted two independent training procedures for two different reference trajectories.). Model-based RL is possible to learn a control law which ensures the closed-loop stability by introducing a Lyapunov constraint into the  objective function of the policy improvement according to the latest research \cite{Berkenkamp2017NIPS}.  However, the model-based RL with stability guarantees requires an admissible control law --- a control law which makes the original system asymptotically stable --- for the initialization. Both the Lyapunov candidate function and complete system dynamics are assumed to be Lipschitz continuous with known Lipschitz constants for the construction of the Lyapunov constraint. It is challenging to find the Lipschitz constant of an uncertain system subject to unknown environmental disturbances. Therefore, the introduced Lyapunov constraint function is restrictive, as it is established based on the worst-case consideration \cite{Berkenkamp2017NIPS}. 

With the consideration of merits and limitations of existing RL methods, we propose a novel learning-based control algorithm for uncertain ASVs by combining a conventional control method with deep RL in this paper. The proposed learning-based control design, therefore, consists of two components: a baseline control law stabilizing a nominal ASV system and a deep RL control law used to compensate for system uncertainties and disturbances. Such a design method has several advantages over both conventional model-based methods and pure deep RL methods. First of all, in relation to the ``model-free'' feature of deep RL, we can learn a control law directly to compensate for uncertainties and disturbances without exploiting their structures, boundaries, or frequencies. In the new design, uncertainties and disturbances are not necessarily matched, as deep RL seeks a control law like direct adaptive control \cite{Sutton1992CSM}. The learning process is performed offline using historical data and the stochastic gradient descent technique, so there is no need for the ASV system be persistently excited when the learned control law is implemented. Second, the overall learned control law can provide stability guarantees, if the baseline control law is able to stabilize the ASV system at least locally. Without introducing a restrictive Lyapunov constraint into the objective function of the policy improvement in RL as in \cite{Berkenkamp2017NIPS}, we can avoid exploiting the Lipschitz constant of the overall system and potentially produce less conservative results. Lastly, the proposed design is potentially more sample efficient than a RL algorithm learning from scratch -- that is, fewer data samples are needed for the training process. In RL, a system learns from mistakes so a lot of trial and error is demanded. Fortunately, in our proposed design, the baseline control which can stabilize the overall system under no disturbances, can help to exclude unnecessary mistakes, so it provides a good starting point for the RL training. A similar idea is used in \cite{Hwangbo2017RAL} for the control of quadrotors. The baseline control in \cite{Hwangbo2017RAL} is constructed based on the full accurate model of a quadrotor system, but stability analysis is missing.

The rest of the paper is organized as follows. In Section \ref{sec:Prelim}, we present the ASV dynamics,  basic concepts of reinforcement learning, and problem formulation. Section \ref{sec:DeepRLControlDesign} describes the proposed methodology, including deep reinforcement learning design, training setup, and  algorithm analysis. In Section \ref{sec:Sim&Exp}, numerical  simulation results are provided to show the efficiency of the proposed design. Conclusion remarks are given in Section \ref{sec:Concl}.

\section{Problem formulation} \label{sec:Prelim}

The full dynamics of autonomous surface vehicles (ASVs) have six degrees of freedom (DOF), including three linear motions and three rotational motions \cite{Fossen2011Handbook}. In most scenarios, we are interested in controlling the horizontal dynamics of (ASVs) \cite{Skjetnea2005Auto,Peng2013}. We, therefore, ignore the vertical, rolling, and pitching motions of ASVs by default in this paper.

Let $x$ and $y$ be the horizontal position coordinates of an ASV in the inertial frame and $\psi$ the heading angle as shown in Figure \ref{fig:ship}. In the body frame (c.f., Figure \ref{fig:ship}), we use $u$ and $v$ to represent the linear velocities in surge ($x$-axis) and sway ($y$-axis), respectively. The heading angular rate is denoted by $r$.  The general 3-DOF nonlinear dynamics of an ASV can be expressed as
\begin{equation}
    \left\{
    \begin{array}{rcl}
       \dot{\boldsymbol{\eta}} &=& \boldsymbol{R}\left(\boldsymbol{\eta}\right)\boldsymbol{\nu}  \\
       \boldsymbol{M}\dot{\boldsymbol{\nu}}+\left(\boldsymbol{C}\left(\boldsymbol{\nu} \right)+\boldsymbol{D}\left(\boldsymbol{\nu} \right)\right)\boldsymbol{\nu} +\boldsymbol{G}\left(\boldsymbol{\nu} \right) &=& \boldsymbol{\tau}
    \end{array}
    \right. \label{eq:ASV_Dyn}
\end{equation}
where $\boldsymbol{\eta}=\left[x, y, \psi\right]^T\in\mathbb{R}^{3}$ is a generalized coordinate vector,  $\boldsymbol{\nu}=\left[u,v,r\right]^T\in\mathbb{R}^{3}$ is the speed vector, $\boldsymbol{M}$ is the inertia matrix, $\boldsymbol{C}\left(\boldsymbol{\nu} \right)$ denotes the
matrix of Coriolis and centripetal terms, $\boldsymbol{D}\left(\boldsymbol{\nu}\right)$ is the damping matrix, $ \boldsymbol{\tau}\in\mathbb{R}^{3}$ represents the control forces and moments, $\boldsymbol{G}\left(\boldsymbol{\nu} \right)=\left[\boldsymbol{g}_{1}\left(\boldsymbol{\nu} \right), \boldsymbol{g}_{2}\left(\boldsymbol{\nu} \right), \boldsymbol{g}_{3}\left(\boldsymbol{\nu} \right)\right]^T\in\mathbb{R}^{3}$ denotes unmodeled dynamics due to gravitational and buoyancy forces and moments \cite{Fossen2011Handbook},  and $\boldsymbol{R}$ is a rotation matrix given by  
\begin{equation*}
    \boldsymbol{R}=\left[\begin{array}{ccc}
    \cos{\psi} & -\sin{\psi} & 0 \\
    \sin{\psi} & \cos{\psi} &  0 \\
    0 & 0 & 1
    \end{array}\right] 
\end{equation*}
\begin{figure}[tbp]
    \centering
    \includegraphics[width=0.375\textwidth]{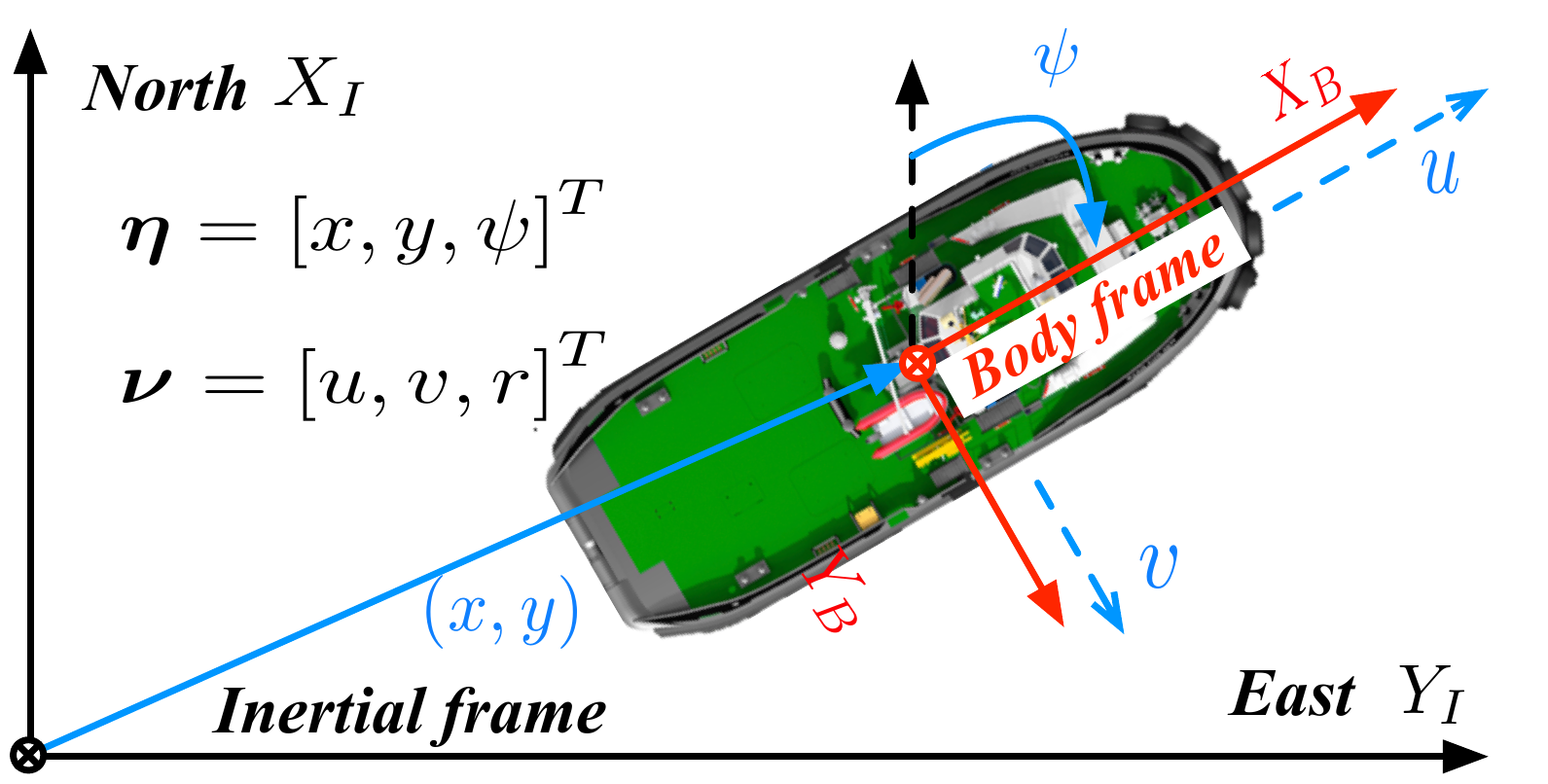}
    \caption{Coordinate systems of an autonomous surface vehicle}
    \label{fig:ship}
\end{figure}

The inertia matrix $\boldsymbol{M}=\boldsymbol{M}^T>0$ is 
\begin{equation}
    \boldsymbol{M}=[M_{ij}]=\left[\begin{array}{ccc}
         M_{11}&   0    & 0  \\
         0     & M_{22} & M_{23} \\
         0     &  M_{32}& M_{33}
    \end{array}\right]
\end{equation}
where $M_{11}=m-X_{\dot{u}}$, $M_{22}=m-Y_{\dot{v}}$,  $M_{33}=I_z-N_{\dot{r}}$, and $M_{32}=M_{23}=mx_g-Y_{\dot{r}}$. The matrix $\boldsymbol{C}\left(\boldsymbol{\nu} \right)=-\boldsymbol{C}^T\left(\boldsymbol{\nu} \right)$ is 
\begin{equation}
    \boldsymbol{C}=[C_{ij}]=\left[\begin{array}{ccc}
         0 &   0    & C_{13}\left(\boldsymbol{\nu} \right) \\
         0 &   0 & C_{23}\left(\boldsymbol{\nu} \right) \\
         -C_{13}\left(\boldsymbol{\nu} \right)     & -C_{23}\left(\boldsymbol{\nu} \right) & 0
    \end{array}\right]
\end{equation}
where $C_{13}\left(\boldsymbol{\nu} \right)=-M_{22}v-M_{23}r$, $C_{23}\left(\boldsymbol{\nu} \right)=-M_{11}u$. The damping matrix $\boldsymbol{D}\left(\boldsymbol{\nu}\right)$ is 
\begin{equation}
    \boldsymbol{D}\left(\boldsymbol{\nu}\right)=[D_{ij}]=
    \left[\begin{array}{ccc}
         D_{11}\left(\boldsymbol{\nu} \right) &   0    & 0\\
         0     & D_{22}\left(\boldsymbol{\nu} \right)& D_{23}\left(\boldsymbol{\nu} \right) \\
         0 & D_{32}\left(\boldsymbol{\nu} \right) & D_{33}\left(\boldsymbol{\nu} \right) 
    \end{array}\right]
\end{equation}
where $D_{11}\left(\boldsymbol{\nu} \right)=-X_u-X_{\vert u\vert u}\vert u\vert -X_{uuu}u^2$, $D_{22}\left(\boldsymbol{\nu} \right)=-Y_v-Y_{\vert v\vert v}\vert v\vert -Y_{\vert r\vert v}\vert r\vert $, $D_{23}\left(\boldsymbol{\nu} \right)=-Y_r-Y_{\vert v\vert r}\vert v\vert -Y_{\vert r\vert r}\vert r\vert $, 
$D_{32}\left(\boldsymbol{\nu} \right)=-N_v-N_{\vert v\vert v}\vert v\vert -N_{\vert r\vert v}\vert r\vert $, $D_{33}\left(\boldsymbol{\nu} \right)=-N_r-N_{\vert v\vert r}\vert v\vert -N_{\vert r\vert r}\vert r\vert $, and $X_{\left(\cdot\right)}$, $Y_{\left(\cdot\right)}$, and $N_{\left(\cdot\right)}$ are hydrodynamic coefficients whose definitions can be found in \cite{Fossen2011Handbook}. Accurate numerical models of the nonlinear dynamics (\ref{eq:ASV_Dyn}) are rarely available. Major uncertainty sources come from $\boldsymbol{M}$, $\boldsymbol{C}\left(\boldsymbol{\nu} \right)$, and $\boldsymbol{D}\left(\boldsymbol{\nu}\right)$ due to hydrodynamics, and  $\boldsymbol{G}\left(\boldsymbol{\nu} \right)$ due to gravitational and buoyancy forces and moments.  The objective of this work is to design a control scheme capable of handling these uncertainties.

\section{Model-Reference Reinforcement Learning Control} \label{sec:MR_DeepRL}
Let $\boldsymbol{x}=\left[\boldsymbol{\eta}^T, \boldsymbol{\nu}^T\right]^T$ and $\boldsymbol{u}=\boldsymbol{\tau}$, so (\ref{eq:ASV_Dyn})  can be rewritten as 
\begin{equation}
    \dot{\boldsymbol{x}} = \left[\begin{array}{cc}
    0 &  \boldsymbol{R}\left(\boldsymbol{\eta}\right) \\
    0 & \boldsymbol{A}\left(\boldsymbol{\nu}\right)
    \end{array}\right]\boldsymbol{x} + \left[\begin{array}{c}
    0 \\
    \boldsymbol{B}
    \end{array}\right]\boldsymbol{u} \label{eq:ASV_Dyn2}
\end{equation}
where $\boldsymbol{A}\left(\boldsymbol{\nu}\right)=\boldsymbol{M}^{-1}\left(\boldsymbol{C}\left(\boldsymbol{\nu} \right)+\boldsymbol{D}\left(\boldsymbol{\nu} \right)\right)$, and $\boldsymbol{B}=\boldsymbol{M}^{-1}$. Assume an accurate model (\ref{eq:ASV_Dyn2}) is not available, but it is possible to get a nominal model expressed as
\begin{equation}
    \dot{\boldsymbol{x}}_m = \left[\begin{array}{cc}
    0 &  \boldsymbol{R}\left(\boldsymbol{\eta}\right) \\
    0 & \boldsymbol{A}_m
    \end{array}\right]\boldsymbol{x}_m + \left[\begin{array}{c}
    0 \\
    \boldsymbol{B}_m
    \end{array}\right]\boldsymbol{u}_m \label{eq:ASV_Dyn_nom}
\end{equation}
where $\boldsymbol{A}_m$ and $\boldsymbol{B}_m$ are the known system matrices.  Assume that there exists a control law $\boldsymbol{u}_m$ allowing the states of the nominal system (\ref{eq:ASV_Dyn_nom}) to converge to a reference signal $\boldsymbol{x}_r$, i.e., $\Vert\boldsymbol{x}_m-\boldsymbol{x}_r\Vert_2\to{0}$ as $t\to\infty$.

The objective is to design a control law allowing the state of (\ref{eq:ASV_Dyn2}) to track state trajectories of the nominal model (\ref{eq:ASV_Dyn_nom}).  As shown in Figure \ref{fig:Cntrlblock}, the overall control law for the ASV system (\ref{eq:ASV_Dyn2}) has the following expression.
\begin{equation}
    \boldsymbol{u} =\boldsymbol{u}_b + \boldsymbol{u}_{l}\label{eq:entireCntrl}
\end{equation}
where $\boldsymbol{u}_b$ is a baseline control designed based on (\ref{eq:ASV_Dyn_nom}), and $\boldsymbol{u}_{l}$ is a control policy from the deep reinforcement learning module shown in Figure \ref{fig:Cntrlblock}. The baseline control $\boldsymbol{u}_b$ is employed to ensure some basic performance, (i.e., local stability), while $\boldsymbol{u}_{l}$ is introduced to compensate for all system uncertainties. The baseline control $\boldsymbol{u}_b$ in (\ref{eq:entireCntrl}) can be designed based on any existing model-based method based on the nominal model (\ref{eq:ASV_Dyn_nom}). Hence, we ignore the design process of $\boldsymbol{u}_b$, and mainly focus on the development of $\boldsymbol{u}_{l}$ based on reinforcement learning. 
\begin{figure}
    \centering
    \includegraphics[width=0.45\textwidth]{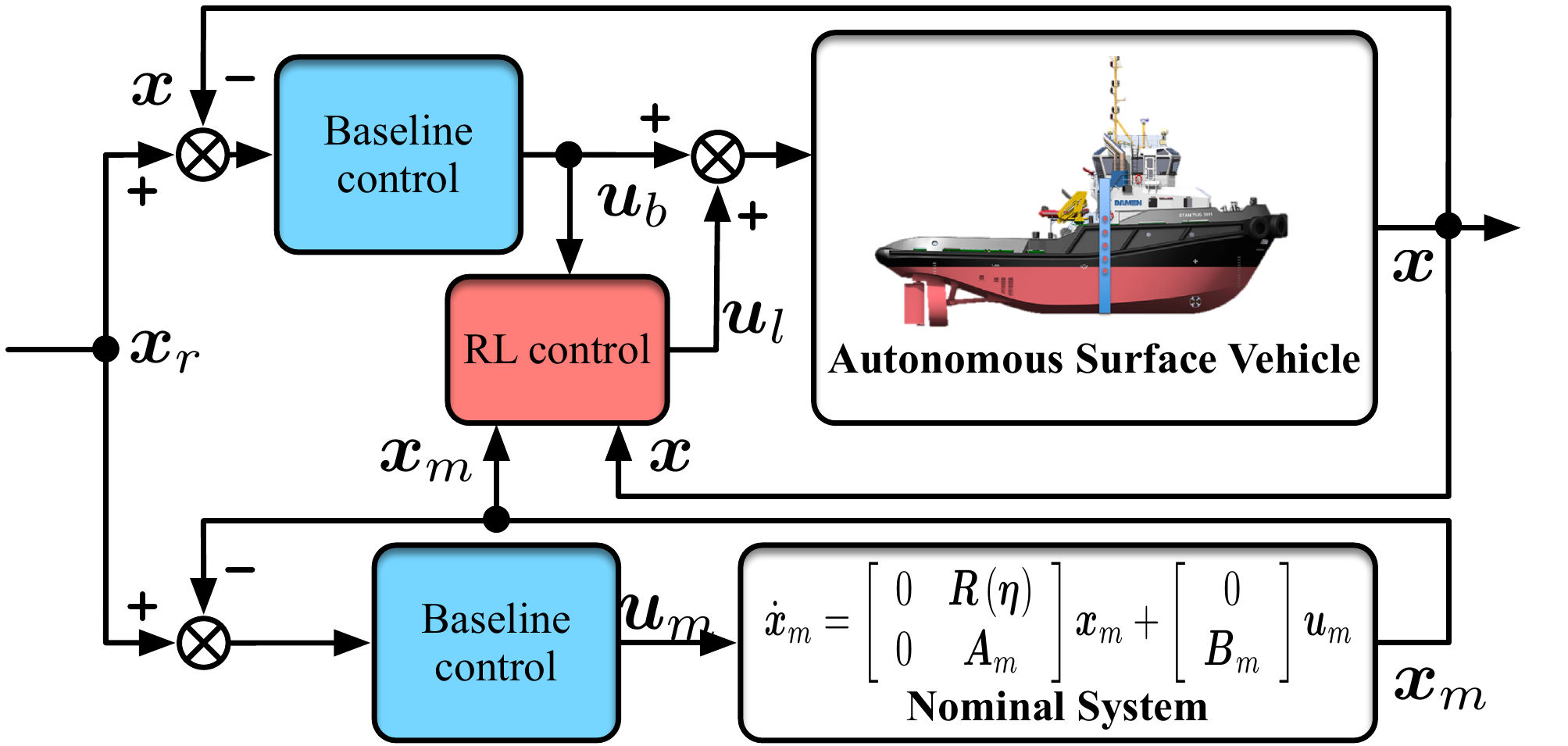}
    \caption{Model-reference reinforcement learning control}
    \label{fig:Cntrlblock}
\end{figure}

\subsection{Reinforcement learning} \label{subsec:RL}
In RL, system dynamics are characterized using a Markov decision process  denoted by a tuple $\mathcal{MDP}:=\big\langle \mathcal{S},\;\mathcal{U},\;\mathcal{P},\;R,\;\gamma\big\rangle$, where $\mathcal{S}$ is the state space, $\mathcal{U}$ specifies the action/input space, $\mathcal{P}:\mathcal{S}\times\mathcal{U}\times\mathcal{S}\rightarrow \mathbb{R}$ defines a transition probability, $R:\mathcal{S}\times\mathcal{U}\rightarrow \mathbb{R}$ is a reward function, and $\gamma\in\left[0,\;1\right]$ is a discount factor. A policy in RL, denoted by $\boldsymbol{\pi}\left(\boldsymbol{u}_l|\boldsymbol{s}\right)$, is the probability of choosing an action $\boldsymbol{u}_l\in\mathcal{U}$ at a state $\boldsymbol{s}\in\mathcal{S}$. Note that the state vector $\boldsymbol{s}$ contains all available signals affecting the reinforcement learning control $\boldsymbol{u}_l$. In this paper, such signals include $\boldsymbol{x}$, $\boldsymbol{x}_{m}$, $\boldsymbol{x}_{r}$, and $\boldsymbol{u}_{b}$, where $\boldsymbol{x}_{m}$ performs like a target state for system (\ref{eq:ASV_Dyn2}) and $\boldsymbol{u}_{b}$ is a function of $\boldsymbol{x}$ and $\boldsymbol{x}_{r}$. Hence, we choose $\boldsymbol{s}=\left\{\boldsymbol{x}_{m},\boldsymbol{x}, \boldsymbol{u}_{b}\right\}$.

Reinforcement learning uses data samples, so it is assumed that we can sample input and state data from system (\ref{eq:ASV_Dyn2}) at discrete time steps. Without loss of generality, we define $\boldsymbol{x}_{t}$, $\boldsymbol{u}_{b,t}$, and $\boldsymbol{u}_{l,t}$ as the ASV state, the baseline control action, and the control action from the reinforcement learning at the time step $t$, respectively. The state signal $\boldsymbol{s}$ at the time step $t$ is, therefore, denoted by  $\boldsymbol{s}_t=\left\{\boldsymbol{x}_{m,t},\boldsymbol{x}_{t}, \boldsymbol{u}_{b,t}\right\}$. The sample time step is assumed to be fixed and denoted by $\delta t$.  

For each state $\boldsymbol{s}_t$, we define a value function $V_{\boldsymbol{\pi}}\left(\boldsymbol{s}_t\right)$ as an expected accumulated return described as
\begin{equation}
V_{\boldsymbol{\pi}}
=\sum_{t}^{\infty}\sum_{\boldsymbol{u}_{l,t}}\boldsymbol{\pi}\left(\boldsymbol{u}_{l,t}|\boldsymbol{s}_t\right)\sum_{\boldsymbol{s}_{t+1}}\mathcal{P}_{t+1|t}\big(R_t+\gamma V_{\boldsymbol{\pi}}(\boldsymbol{s}_{t+1}) \big) 
\label{eq:V_Func}
\end{equation}
where $R_{t}=R(\boldsymbol{s}_t,\boldsymbol{u}_{l,t})$ and $\mathcal{P}_{t+1|t}=\mathcal{P}\left(\boldsymbol{s}_{t+1}\left|\boldsymbol{s}_t,\boldsymbol{u}_{l,t}\right.\right)$. The action-value function (a.k.a., Q-function) is defined to be 
\begin{equation}
Q_{\boldsymbol{\pi}}\left(\boldsymbol{s}_t,\boldsymbol{u}_{l,t}\right)=R_t+\gamma \sum_{\boldsymbol{s}_{t+1}}\mathcal{P}_{t+1|t}V_{\boldsymbol{\pi}}(\boldsymbol{s}_{t+1}) \label{eq: Action-Value Func}
\end{equation}
In our design, we aim to allow system (\ref{eq:ASV_Dyn2}) to track the nominal system (\ref{eq:ASV_Dyn_nom}), so $R_t$ is defined as
\begin{equation}
   R_t = -\left(\boldsymbol{x}_{t}-\boldsymbol{x}_{m,t}\right)^T\boldsymbol{G} \left(\boldsymbol{x}_{t}-\boldsymbol{x}_{m,t}\right)-\boldsymbol{u}_{l, t}^T\boldsymbol{H}\boldsymbol{u}_{l, t} \label{eq:ASV_Reward}
\end{equation}
where $\boldsymbol{G} \geq 0$ and $\boldsymbol{H}>0$ are positive definite matrices.

The objective of the reinforcement learning is to find an optimal policy $\boldsymbol{\pi}*$ to maximize the state-value function $V_\pi(s_t)$ or the action-value function $Q_{\boldsymbol{\pi}}\left(\boldsymbol{s}_t,\;\boldsymbol{u}_{l,t}\right)$, $\forall \boldsymbol{s}_t\in\mathcal{S} $, namely, 
\begin{align}
    \boldsymbol{\pi}^*&=\arg\max_{\boldsymbol{\pi}}Q_{\boldsymbol{\pi}}\left(\boldsymbol{s}_t,\boldsymbol{u}_{l, t}\right) \nonumber \\
    &= \arg\max_{\boldsymbol{\pi}} \left(R_t+\gamma \sum_{\boldsymbol{s}_{t+1}}\mathcal{P}_{t+1|t}V_{\boldsymbol{\pi}}(\boldsymbol{s}_{t+1})\right)\label{eq:RL_Obj}
\end{align}

\section{Deep Reinforcement Learning Control Design} \label{sec:DeepRLControlDesign}
In this section, we will present a deep reinforcement learning algorithm for the design of  $\boldsymbol{u}_{l}$ in (\ref{eq:entireCntrl}), where both the control law $\boldsymbol{u}_{l}$ and the Q-function $Q_{\boldsymbol{\pi}}\left(\boldsymbol{s}_t,\boldsymbol{u}_{l, t}\right)$ are approximated using deep neural networks.  


The deep reinforcement learning control in this paper is developed based on the soft actor-critic (SAC) algorithm which provides both sample efficient learning and convergence \cite{Haarnoja2018SAC1}. 
In SAC, an entropy term is added to the objective function in (\ref{eq:RL_Obj}) to regulate the exploration performance at the training stage. The objective of  (\ref{eq:RL_Obj}) is thus rewritten as
\begin{align}
     \boldsymbol{\pi}^*=& \arg\max_{\boldsymbol{\pi}} \left(R_t+\gamma \mathbb{E}_{\boldsymbol{s}_{t+1}}\left[V_{\boldsymbol{\pi}}(\boldsymbol{s}_{t+1}) \right.\right.\nonumber \\
     &\big.\left.+\alpha\mathcal{H}\left(\boldsymbol{\pi}\left(\boldsymbol{u}_{l,t+1}|\boldsymbol{s}_{t+1}\right)\right)\right]\big)\label{eq:SAC_Obj}
\end{align}
where  $\mathbb{E}_{\boldsymbol{s}_{t+1}}\left[\cdot\right]=\sum_{\boldsymbol{s}_{t+1}}\mathcal{P}_{t+1|t}\left[\cdot\right]$ is an expectation operator,  $\mathcal{H}\left(\boldsymbol{\pi}\left(\boldsymbol{u}_{l,t}|\boldsymbol{s}_{t}\right)\right)=-\sum_{\boldsymbol{u}_{l,t}}\boldsymbol{\pi}\left(\boldsymbol{u}_{l,t}|\boldsymbol{s}_t\right)\ln\left(\boldsymbol{\pi}\left(\boldsymbol{u}_{l,t}|\boldsymbol{s}_{t}\right)\right)=-\mathbb{E}_{\boldsymbol{\pi}}\left[\ln\left(\boldsymbol{\pi}\left(\boldsymbol{u}_{l,t}|\boldsymbol{s}_{t}\right)\right)\right]$ is the entropy of the policy, and $\alpha$ is a temperature parameter. 


Training of SAC repeatedly executes policy evaluation and policy improvement. In the policy evaluation, a soft Q-value is computed by applying a Bellman operation $Q_{\boldsymbol{\pi}}\left(\boldsymbol{s}_t,\boldsymbol{u}_{l, t}\right)=\mathcal{T}^{\boldsymbol{\pi}}Q_{\boldsymbol{\pi}}\left(\boldsymbol{s}_t,\boldsymbol{u}_{l, t}\right)$ where
\begin{align}
    \mathcal{T}^{\boldsymbol{\pi}}Q_{\boldsymbol{\pi}}\left(\boldsymbol{s}_t,\boldsymbol{u}_{l, t}\right)&=R_t+\gamma \mathbb{E}_{\boldsymbol{s}_{t+1}}\left\{\mathbb{E}_{\boldsymbol{\pi}}\left[Q_{\boldsymbol{\pi}}\left(\boldsymbol{s}_{t+1},\boldsymbol{u}_{l,t+1}\right) \right. \right. \nonumber \\
    &\left.\left.-\alpha\ln\left(\boldsymbol{\pi}\left(\boldsymbol{u}_{l,t+1}|\boldsymbol{s}_{t+1}\right)\right)\right]\right\} \label{eq:BellmanOp}
\end{align}
In the policy improvement, the policy is updated by 
\begin{equation}
    \boldsymbol{\pi}_{new} = \arg \min_{\boldsymbol{\pi}'}\mathscr{D}_{KL}\left(\boldsymbol{\pi}'\left(\cdot\vert\boldsymbol{s}_{t}\right) \Big\Vert {Z^{{\boldsymbol{\pi}}_{ old}}}{e^{Q^{{\boldsymbol{\pi}}_{ old}}\left(\boldsymbol{s}_{t}, \cdot\right)}}\right) \label{eq:KL_pi}
\end{equation}
where $\boldsymbol{\pi}_{ old}$ denotes the policy from the last update, $Q^{{\boldsymbol{\pi}}_{ old}}$ is the Q-value of $\boldsymbol{\pi}_{ old}$. $\mathscr{D}_{KL}$ denotes the Kullback-Leibler (KL) divergence, and $Z^{{\pi}_{old}}$ is a normalization factor. Via mathematical manipulations, the objective for the policy improvement is transformed into 
\begin{equation}
    \boldsymbol{\pi}* = \arg \min_{\boldsymbol{\pi}} \mathbb{E}_{{\boldsymbol{\pi}}}\Big[\alpha\ln\left(\boldsymbol{\pi}\left(\boldsymbol{u}_{l,t}|\boldsymbol{s}_{t}\right)\right)-Q\left(\boldsymbol{s}_{t}, \boldsymbol{u}_{l, t}\right)\Big] \label{eq:PI_Q}
\end{equation}
More details on how (\ref{eq:PI_Q}) is obtained can be found in \cite{Haarnoja2018SAC1, Haarnoja2018SAC2}. As shown in Figure \ref{fig:ACNN}, both the policy $\boldsymbol{\pi}\left(\boldsymbol{u}_{l,t}|\boldsymbol{s}_{t}\right)$ and value function $Q_{\boldsymbol{\pi}}\left(\boldsymbol{s}_{t}, \boldsymbol{u}_{l, t}\right)$ will be parameterized using fully connected multiple layer perceptrons (MLP) with 'ReLU' nonlinearities as the activation functions. The 'ReLU' function is defined as 
\begin{equation*}
    \underline{relu}\left(z\right) = \max\left\{z,0\right\}
\end{equation*}
The ``ReLU'' activation function outperforms other activation functions like sigmoid functions \cite{Dahl2013ICASSP}. For a vector $z=[z_1,\ldots, z_n]^T\in\mathbb{R}^{n}$, there exists $\underline{relu}\left(z\right)=[\underline{relu}\left(z_1\right),\ldots, \underline{relu}\left(z_n\right)]^T$. Hence, a MLP with 'ReLU' as the activation functions and one hidden layer is expressed as 
\begin{equation*}
\underline{MLP}\left(z\right) = W_1 \left[\underline{relu}\left(W_0\left[z^T,\; 1\right]\right)^T,\; 1\right]^T
\end{equation*}
where $\left[z^T,\; 1\right]^T$ is a vector composed of $z$ and $1$, and $W_0$ and $W_1$ with appropriate dimensions are weight matrices to be trained. For the simplicity, we use  $W=\left\{W_0,\;W_1\right\}$ to represent the set of parameters to be trained.
\begin{figure}[tbp]
    \centering
    \includegraphics[width=0.325\textwidth]{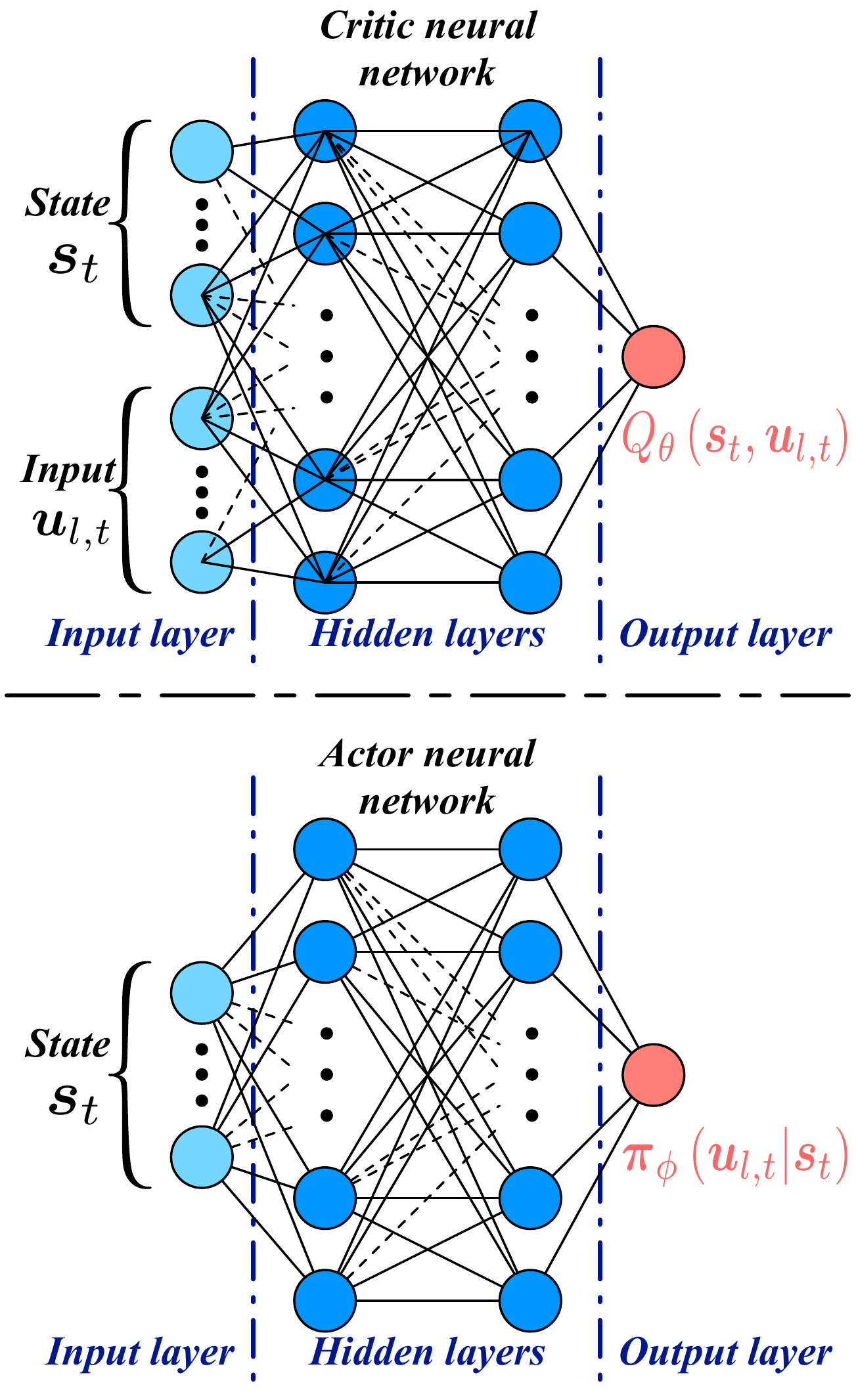}
    \caption{Approximation of $Q_{\theta}$ and $\boldsymbol{\pi}_{{\phi}}$ using MLP}
    \label{fig:ACNN}
\end{figure}

In this paper,  the Q-function is parameterized using $\theta$ and denoted by  $Q_{\theta}\left(\boldsymbol{s}_{t}, \boldsymbol{u}_{l, t}\right)$. The parameterized policy is denoted by $\boldsymbol{\pi}_{{\phi}}\left(\boldsymbol{u}_{l, t}\vert \boldsymbol{s}_{t}\right)$, where $\phi$ is the parameter set to be trained. Note that both $\theta$ and $\phi$ are a set of parameters whose dimensions are determined by the deep neural network setup. For example, if $Q_{\theta}$ is represented by a MLP with $K$ hidden layers and $L$ neurons for each hidden layers, the parameter set $\theta$ is $\theta =\left\{\theta_0, \theta_1, \ldots, \theta_K\right\}$ with $\theta_0\in\mathbb{R}^{\left(dim_{\boldsymbol{s}}+dim_{\boldsymbol{u}}+1\right)\times L} $, $\theta_K\in\mathbb{R}^{\left(L+1\right)}$, and $\theta_i\in\mathbb{R}^{1\times{\left(L\right)}\times {\left(L+1\right)}}$ for $1\leq i\leq K-1$, where $dim_{\boldsymbol{s}}$ denotes the dimension of the state $\boldsymbol{s}$ and $dim_{\boldsymbol{u}}$ is the dimension of the input $\boldsymbol{u}_l$.  The deep neural network for $Q_{\theta}$ is called critic, while the one for $\boldsymbol{\pi}_{{\phi}}$ is called actor.

\begin{figure}[tbp]
    \centering
    \includegraphics[width=0.4\textwidth]{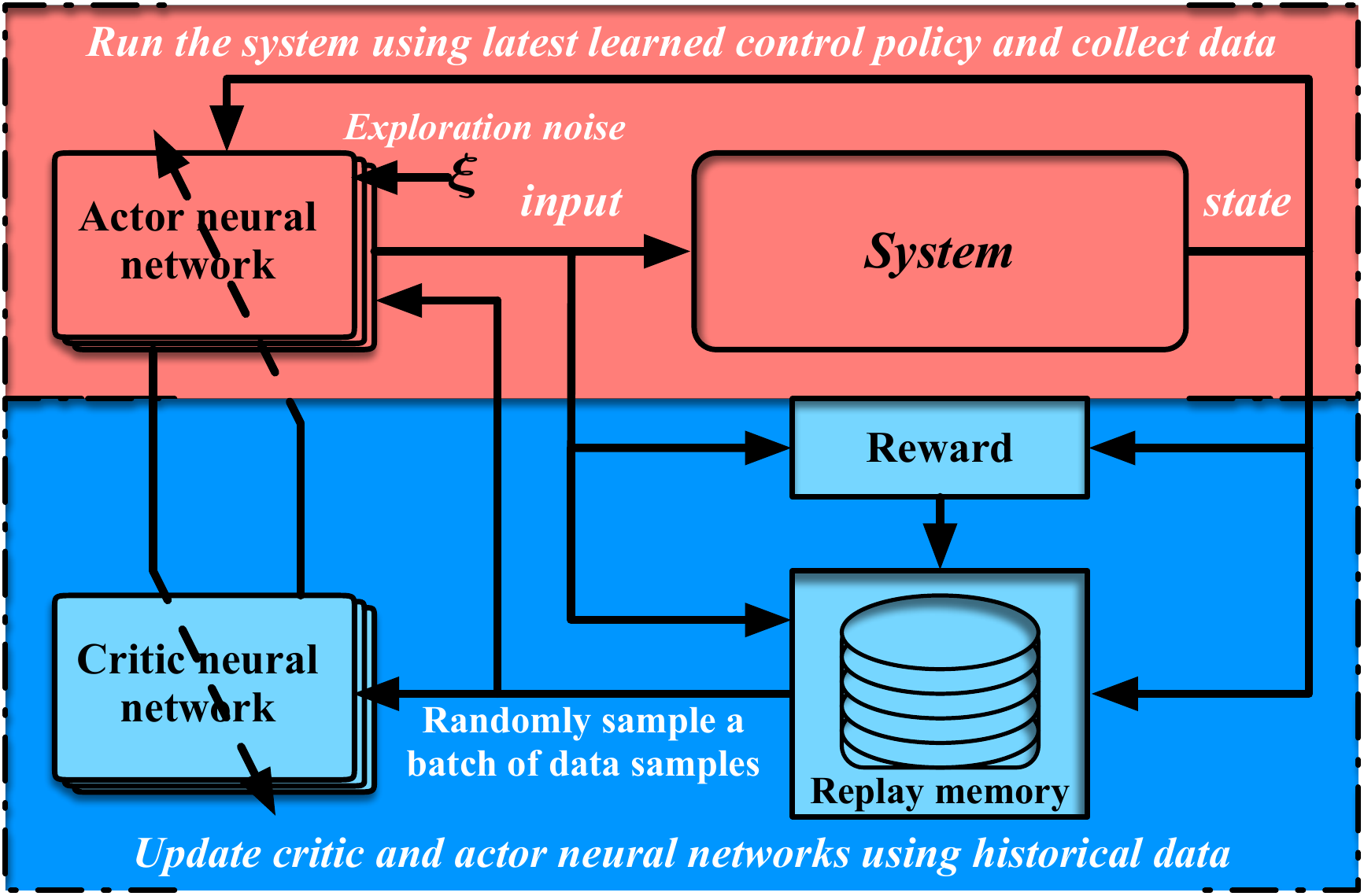}
    \caption{Offline training process of deep reinforcement learning}
    \label{fig:TrainRL}
\end{figure}
\subsection{Training setup}
The algorithm training process is illustrated in Figure  \ref{fig:TrainRL}. The whole training process will be offline. We repeatedly run the system (\ref{eq:ASV_Dyn2}) under a trajectory tracking task. At each time step $t+1$, we collect data samples, such as an input from the last time step  $\boldsymbol{u}_{l, t}$, a state from the last time step $\boldsymbol{s}_{t}$, a reward $R_t$, and a current state $\boldsymbol{s}_{t+1}$. Those historical data will be stored as a tuple $\left(\boldsymbol{s}_{t}, \boldsymbol{u}_{l, t}, R_t, \boldsymbol{s}_{t+1}\right)$ at a replay memory $\mathcal{D}$ \cite{Mnih2015Nature}. At each policy evaluation or improvement step, we randomly sample a batch of historical data, $\mathcal{B}$, from the replay memory $\mathcal{D}$ for the training of the parameters $\theta$ and $\phi$. Starting the training, we apply the baseline control policy $\boldsymbol{u}_b$ to an ASV system to collect the initial data  $\mathcal{D}_0$ as shown in Algorithm \ref{alg:SAC_Control}. The initial data set $\mathcal{D}_0$ is used for the initial fitting of Q-value functions. When the initialization is over, we execute both $\boldsymbol{u}_b$ and the latest updated reinforcement learning policy $\boldsymbol{\pi}_{{\phi}}\left(\boldsymbol{u}_{l, t}\vert \boldsymbol{s}_{t}\right)$ to run the ASV system.

At the policy evaluation step, the parameters $\theta$ are trained to minimize the following Bellman residual.
\begin{equation}
J_{Q}\left(\theta\right)= \mathbb{E}_{\left(\boldsymbol{s}_{t}, \boldsymbol{u}_{l, t}\right)\sim\mathcal{D}}\left[\frac{1}{2}\left(Q_{\theta}\left(\boldsymbol{s}_{t}, \boldsymbol{u}_{l, t}\right)-Y_{target}\right)^2\right] \label{eq:JQ_theta}
\end{equation}
where $\left(\boldsymbol{s}_{t}, \boldsymbol{u}_{l, t}\right)\sim\mathcal{D}$ implies that we randomly pick data samples $\left(\boldsymbol{s}_{t}, \boldsymbol{u}_{l, t}\right)$ from a replay memory $\mathcal{D}$, and 
\begin{equation*}
  Y_{target}=R_t+\gamma\mathbb{E}_{\boldsymbol{s}_{t+1}}\big[\mathbb{E}_{\boldsymbol{\pi}}\left[Q_{\bar{\theta}}\left(\boldsymbol{s}_{t+1}, \boldsymbol{u}_{l,t+1}\right)-\alpha\ln\left(\boldsymbol{\pi}_{\phi}\right)\right] \big]  
\end{equation*} 
where $\bar{\theta}$ is the target parameter which will be updated slowly. Applying a stochastic gradient descent technique (ADAM \cite{Kingma2014Adam} in this paper) to (\ref{eq:JQ_theta}) on a data batch $\mathcal{B}$ with a fixed size, we obtain 
\begin{align*}
    \nabla_{\theta} J_{Q}\left(\theta\right) &= \sum  \frac{\nabla_{\theta}Q_{\theta}}{\vert\mathcal{B}\vert} \Big(Q_{\theta}\left(\boldsymbol{s}_{t}, \boldsymbol{u}_{l, t}\right)-Y_{target}\Big) 
\end{align*}
where $\vert\mathcal{B}\vert$ is the batch size. 

At the policy improvement step, the objective function defined in (\ref{eq:PI_Q}) is represented  using data samples from the replay memory $\mathcal{D}$ as given in (\ref{eq:PI_Phi}).  
\begin{align}
    J_{{\pi}}\left(\phi\right)&=\mathbb{E}_{\left(\boldsymbol{s}_{t}, \boldsymbol{u}_{l, t}\right)\sim\mathcal{D}}\Big(\alpha\ln(\pi_{{\phi}})  -Q_{{\theta}}\left(\boldsymbol{s}_{t}, \boldsymbol{u}_{l, t}\right)\Big) \label{eq:PI_Phi}
\end{align}
Parameter $\phi$ is trained to minimize (\ref{eq:PI_Phi})  using a stochastic gradient descent technique.  At the training stage, the actor neural network is expressed as
\begin{equation}
    {\boldsymbol{u}}_{l,\phi} = \bar{\boldsymbol{u}}_{l,\phi}+\boldsymbol{\sigma}_{\phi}\odot\boldsymbol{\xi} \label{eq:ActorNN}
\end{equation}
where  $\bar{\boldsymbol{u}}_{l,\phi}$ represents the control law to be implemented in the end, $\boldsymbol{\sigma}_{\phi}$ denotes the standard deviation of the exploration noise, $\boldsymbol{\xi}\sim \mathscr{N}\left(0, \boldsymbol{I}\right)$ is the exploration noise with $\mathscr{N}\left(0, \boldsymbol{I}\right)$ denoting a Gaussian distribution, and ``$\odot$'' is the Hadamard product. Note that the exploration noise $\boldsymbol{\xi}$ is only applied to the training stage. Once the training is done, we only need $\bar{\boldsymbol{u}}_{l,\phi}$ in the implementation. Hence, at the training stage, $u_l$ in Figure \ref{fig:Cntrlblock} is equal to ${\boldsymbol{u}}_{l,\phi}$. Once the training is over, we have $u_l=\bar{\boldsymbol{u}}_{l,\phi}$.

\begin{algorithm}[tbp]
   \caption{Reinforcement learning control} \label{alg:SAC_Control}
\begin{algorithmic}[1]
   \STATE Initialize parameters $\theta_1$, $\theta_2$ for $Q_{{\theta}_{1}}$ and $Q_{{\theta}_{2}}$, respectively, and $\phi$ for the actor network  (\ref{eq:ActorNN}).
   \STATE Assign values to the the target parameters $\bar{\theta}_1\leftarrow\theta_1$, $\bar{\theta}_2\leftarrow\theta_2$, $\mathcal{D}\leftarrow\emptyset$,  
   $\mathcal{D}_0\leftarrow\emptyset$,
  \STATE \textcolor{blue}{Get data set $\mathcal{D}_0$ by running $\boldsymbol{u}_b$ on (\ref{eq:ASV_Dyn2}) with $\boldsymbol{u}_l=\boldsymbol{0}$}
  \STATE \textcolor{blue}{Turn off the exploration and train initial critic parameters $\theta_1^0$, $\theta_2^0$ using $\mathcal{D}_0$ according to (\ref{eq:JQ_theta}). }
  \STATE \textcolor{blue}{Initialize the replay memory $\mathcal{D}\leftarrow\mathcal{D}_0$}
  \STATE \textcolor{blue}{Assign initial values to critic parameters $\theta_1\leftarrow\theta_1^0$, $\theta_2\leftarrow\theta_2^0$ and their targets $\bar{\theta}_1\leftarrow\theta_1^0$, $\bar{\theta}_2\leftarrow\theta_2^0$}
   \REPEAT
   \FOR{each data collection step}
    \STATE Choose an action $\boldsymbol{u}_{l, t}$ according to $ \boldsymbol{\pi}_{\phi}\left(\boldsymbol{u}_{l, t}\vert \boldsymbol{s}_{t}\right)$ 
    \STATE Run both the nominal system (\ref{eq:ASV_Dyn_nom}) and the full system (\ref{eq:ASV_Dyn2}) \& collect $\boldsymbol{s}_{t+1}=\left\{\boldsymbol{x}_{t+1}, \boldsymbol{x}_{m, t+1}, \boldsymbol{u}_{b,t+1}\right\}$
    \STATE $\mathcal{D}\leftarrow\mathcal{D}\bigcup \left\{\boldsymbol{s}_{t}, \boldsymbol{u}_{l, t}, R\left(\boldsymbol{s}_{t}, \boldsymbol{u}_{l, t}\right), \boldsymbol{s}_{t+1}\right\}$
    \ENDFOR
    \FOR{each gradient update step}
    \STATE Sample a batch of data $\mathcal{B}$ from $\mathcal{D}$
   \STATE $\theta_j\leftarrow\theta_j-\iota_Q \nabla_{\theta} J_{Q}\left(\theta_j\right) $,   and $j=1$, $2$
   \STATE $\phi\leftarrow\phi-\iota_\pi \nabla_{\phi}J_{\boldsymbol{\pi}}\left(\phi\right)$, 
   \STATE $\alpha\leftarrow \alpha - \iota_\alpha \nabla_{\alpha}J_{{\alpha}}\left(\alpha\right)$
   \STATE $\bar{\theta}_j\leftarrow\kappa\theta_j+\left(1-\kappa\right)\bar{\theta}_j$,  and $j=1$, $2$
   \ENDFOR
   \UNTIL{convergence (i.e. $J_{Q}\left(\theta\right)<$ a small threshold)}
\end{algorithmic}
\end{algorithm}
Applying the policy gradient technique  to (\ref{eq:PI_Phi}), we can calculate the gradient of $J_{\boldsymbol{\pi}}\left(\phi\right)$ with respect to $\phi$ in terms of the stochastic gradient method as in (\ref{eq:StochG_Pi})
\begin{equation}
    \nabla_{\phi}J_{{\pi}} = \sum \frac{\alpha\nabla_{\phi}\ln\boldsymbol{\pi}_{\phi}+\left(\alpha\nabla_{\boldsymbol{u}_{l}}\ln\boldsymbol{\pi}_{\phi}-\nabla_{\boldsymbol{u}_{l}}Q_{{\theta}}\right)\nabla_{\phi}{\boldsymbol{u}}_{l,\phi}}{\vert\mathcal{B}\vert}\label{eq:StochG_Pi}
\end{equation}
The temperature parameters $\alpha$ are updated by minimizing the following objective function.
\begin{equation}
    J_{\alpha} =\mathbb{E}_{\boldsymbol{\pi}}\left[-\alpha\ln \boldsymbol{\pi}\left(\boldsymbol{u}_{l, t}\vert \boldsymbol{s}_{t}\right)-\alpha\bar{\mathcal{H}}\right]
\end{equation}
where $\bar{\mathcal{H}}$ is a target entropy. Following the same setting in \cite{Haarnoja2018SAC2}, we choose $\bar{\mathcal{H}}=-3$ where ``3'' here represents the action dimension. In the final implementation, we use two critics which are parameterized by $\theta_1$ and $\theta_2$, respectively. The two critics are introduced to reduce the over-estimation issue in the training of critic neural networks \cite{Fujimoto2018TD3}. Under the two-critic mechanism, the target value $Y_{target}$ is
\begin{align}
  Y_{target} &=R_t+\gamma\min\Big\{Q_{\bar{\theta}_{1}}\left(\boldsymbol{s}_{t+1}, \boldsymbol{u}_{l,t+1}\right), \Big.\nonumber \\
  &\Big.Q_{\bar{\theta}_{2}}\left(\boldsymbol{s}_{t+1}, \boldsymbol{u}_{l,t+1}\right)\Big\}-\gamma\alpha\ln\left(\boldsymbol{\pi}_{\phi}\right)  \label{eq:TargetNew}
\end{align}
The entire algorithm is summarized in Algorithm \ref{alg:SAC_Control}. In Algorithm \ref{alg:SAC_Control}, $\iota_Q$, $\iota_\pi$, and $\iota_\alpha$ are positive learning rates (scalars), and $\kappa>0$ is a constant scalar.


\begin{algorithm}[tb]
  \caption{Policy iteration technique} \label{alg:PI_Tech}
\begin{algorithmic}[1]
  \STATE Start from an initial control policy $\boldsymbol{u}_0$
  \REPEAT
  \FOR{Policy evaluation}
    \STATE Under a fixed policy $\boldsymbol{u}_l$, apply the Bellman backup operator $\mathcal{T}^{\pi}$ to the Q value function, $Q\left(\boldsymbol{s}_{t},\boldsymbol{u}_{l,t}\right)=\mathcal{T}^{\pi}Q\left(\boldsymbol{s}_{t},\boldsymbol{u}_{l,t}\right)$ (c.f., (\ref{eq:BellmanOp}))
    \ENDFOR
    \FOR{Policy improvement}
    \STATE Update policy $\boldsymbol{\pi}$ according to (\ref{eq:SAC_Obj})
  \ENDFOR
  \UNTIL{convergence}
\end{algorithmic}
\end{algorithm}
\section{Performance analysis}
In this subsection, both the convergence and stability of the proposed learning-based control are analyzed. For the analysis, the soft actor-critic RL method in Algorithm \ref{alg:SAC_Control} is recapped as a policy iteration (PI) technique which is summarized in Algorithm \ref{alg:PI_Tech}. We thereafter present the following two lemmas without proofs for the convergence analysis \cite{Haarnoja2018SAC1,Haarnoja2018SAC2}.
\begin{lemma}[Policy evaluation] \label{lem:Pi_Eval}
Let $\mathcal{T}^{\pi}$ be the Bellman backup operator under a fixed policy $\boldsymbol{\pi}$ and $Q^{k+1}\left(\boldsymbol{s},\boldsymbol{u}_l\right)=\mathcal{T}^{\pi}Q^{k}\left(\boldsymbol{s},\boldsymbol{u}_l\right)$. The sequence $Q^{k+1}\left(\boldsymbol{s},\boldsymbol{u}_l\right)$ will converge to the soft Q-function $Q^{\boldsymbol{\pi}}$ of the policy $\boldsymbol{\pi}$ as $k\to\infty$.
\end{lemma}
\begin{lemma}[Policy improvement] \label{lem:Pi_Improve}
Let $\boldsymbol{\pi}_{old}$ be an old policy and $\boldsymbol{\pi}_{new}$ be a new policy obtained according to (\ref{eq:KL_pi}). There exists $Q^{\boldsymbol{\pi}_{new}}\left(\boldsymbol{s},\boldsymbol{u}_l\right)\geq Q^{\boldsymbol{\pi}_{old}}\left(\boldsymbol{s},\boldsymbol{u}_l\right)$ $\forall \boldsymbol{s}\in\mathcal{S}$ and $\forall \boldsymbol{u}\in\mathcal{U}$.
\end{lemma}

In terms of (\ref{lem:Pi_Eval}) and (\ref{lem:Pi_Improve}), we are ready to present Theorem \ref{thm:Converge} to show the convergence of the SAC algorithm.
\begin{theorem}[\textbf{Convergence}] \label{thm:Converge}
If one repeatedly applies the policy evaluation and policy improvement steps to any control policy $\boldsymbol{\pi}$, the control policy $\boldsymbol{\pi}$ will converge to an optimal policy $\boldsymbol{\pi}^{*}$ such that  $Q^{\boldsymbol{\pi}^{*}}\left(\boldsymbol{s},\boldsymbol{u}_l\right) \geq Q^{\boldsymbol{\pi}}\left(\boldsymbol{s},\boldsymbol{u}_l\right)$ $\forall \boldsymbol{\pi}\in\Pi$,  $\forall \boldsymbol{s}\in\mathcal{S}$,  and $\forall \boldsymbol{u}\in\mathcal{U}$,  where $\Pi$ denotes a policy set.
\end{theorem}
\begin{proof}
Let $\boldsymbol{\pi}_{i}$ be the policy obtained from the $i$-th policy improvement with $i=0$, $1$, $\ldots$, $\infty$. According to Lemma \ref{lem:Pi_Improve}, one has $Q^{\boldsymbol{\pi}_{i}}\left(\boldsymbol{s},\boldsymbol{u}_l\right)\geq Q^{\boldsymbol{\pi}_{i-1}}\left(\boldsymbol{s},\boldsymbol{u}_l\right)$, so $Q^{\boldsymbol{\pi}_{i}}\left(\boldsymbol{s},\boldsymbol{u}_l\right)$ is monotonically non-decreasing with respect to the policy iteration step $i$. In addition, $Q^{\boldsymbol{\pi}_{i}}\left(\boldsymbol{s},\boldsymbol{u}_l\right)$ is upper bounded according to the definition of the reward given in (\ref{eq:ASV_Reward}), so $Q^{\boldsymbol{\pi}_{i}}\left(\boldsymbol{s},\boldsymbol{u}_l\right)$ will converge to an upper limit $Q^{\boldsymbol{\pi}^{*}}\left(\boldsymbol{s},\boldsymbol{u}_l\right)$ with ${Q}^{\boldsymbol{\pi}^{*}}\left(\boldsymbol{s},\boldsymbol{u}_l\right) \geq Q^{\boldsymbol{\pi}}\left(\boldsymbol{s},\boldsymbol{u}_l\right)$ $\forall \boldsymbol{\pi}\in\Pi$,  $\forall \boldsymbol{s}\in\mathcal{S}$,  and $\forall \boldsymbol{u}_l\in\mathcal{U}$. 
\end{proof}

Theorem \ref{thm:Converge} demonstrates that we can find an optimal policy by repeating the policy evaluation and improvement processes. Next, we will show the closed-loop stability of the overall control law (baseline control $\boldsymbol{u}_b$ plus the learned control $\boldsymbol{u}_{l}$). The following assumption is made for the baseline control developed using the nominal system (\ref{eq:ASV_Dyn_nom}).  
\begin{assumption}\label{assump:BaselineC}
The baseline control law $\boldsymbol{u}_b$ can ensure that the overall uncertain ASV system is stable -- that is, there exists a Lyapunov function $\mathbb{V}\left(\boldsymbol{s}_{t}\right)$ associate with $\boldsymbol{u}_b$ such that $\mathbb{V}\left(\boldsymbol{s}_{t+1}\right)-\mathbb{V}\left(\boldsymbol{s}_{t}\right)\leq 0$ $\forall \boldsymbol{s}_{t} \in \mathcal{S}$.
\end{assumption}
Note that the baseline control $\boldsymbol{u}_b$ is implicitly included in the state vector $\boldsymbol{s}$, as $\boldsymbol{s}$ consists of $\boldsymbol{x}$, $\boldsymbol{x}_m$, and $\boldsymbol{u}_b$ in this paper as discussed in Section \ref{sec:MR_DeepRL}. Hence, $\mathbb{V}\left(\boldsymbol{s}_{t}\right)$ in Assumption \ref{assump:BaselineC} is the Lyapunov function for the closed-loop system of (\ref{eq:ASV_Dyn2}) with the baseline control $\boldsymbol{u}_b$.

Assumption \ref{assump:BaselineC} is possible in real world. One could treat the nominal model (\ref{eq:ASV_Dyn_nom}) as a linearized model of the overall ASV system (\ref{eq:ASV_Dyn2}) around a certain equilibrium. Therefore, a control law, which ensures asymptotic stability for (\ref{eq:ASV_Dyn_nom}), can ensure at least local stability for (\ref{eq:ASV_Dyn2}) \cite{Khalil2002Book}. In the stability analysis, we will ignore the entropy term $\mathcal{H}\left(\boldsymbol{\pi}\right)$, as it will converge to zero in the end and it is only introduced to regulate the exploration magnitude. Now, we present Theorem \ref{thm:Stab} to demonstrate the closed-loop stability of the ASV system (\ref{eq:ASV_Dyn2}) under the composite control law (\ref{eq:entireCntrl}).
\begin{theorem}[\textbf{Stability}] \label{thm:Stab}
Suppose Assumption \ref{assump:BaselineC} holds. The overall control law $\boldsymbol{u}^i = \boldsymbol{u}_b + \boldsymbol{u}_{l}^i$ can always stabilize the ASV system (\ref{eq:ASV_Dyn2}), where $\boldsymbol{u}_{l}^i$ represents the RL control law from $i$-th iteration, and $i=0$, $1$, $2$, ... $\infty$.
\end{theorem}
\begin{proof}
In our proposed algorithm, we start the training/learning using the baseline control law $\boldsymbol{u}_b$. According to Lemma \ref{lem:Pi_Eval}, we are able to obtain the corresponding Q value function for the baseline control law $\boldsymbol{u}_b$. Let the Q value function be $ Q^{0}\left(\boldsymbol{s},\boldsymbol{u}_l\right)$ where $\boldsymbol{u}_l$ is a function of $\boldsymbol{s}$.  According to the definitions of the reward function in (\ref{eq:ASV_Reward}) and Q value function in (\ref{eq: Action-Value Func}), we can choose the Lyapunov function candidate as $\mathbb{V}^{0}\left(\boldsymbol{s}\right) = - Q^{0}\left(\boldsymbol{s},\boldsymbol{u}_l\right)$. If Assumption \ref{assump:BaselineC} holds, there exists $\mathbb{V}^{0}\left(\boldsymbol{s_{t+1}}\right)-\mathbb{V}^{0} \left(\boldsymbol{s_{t}}\right)\leq 0$ $\forall \boldsymbol{s_{t}}\in\mathcal{S}$.

In the policy improvement, the control law is updated by
\begin{equation}
   \boldsymbol{u}^{1} = \min_{\boldsymbol{\pi}} \left(-R_t+\gamma \mathbb{V}^{0}\left(\boldsymbol{s_{t+1}}\right)\right) \label{eq: iter1_min}
\end{equation}
where the expectation operator is ignored as the system is deterministic. For any nonlinear system $\boldsymbol{s}_{t+1}=\boldsymbol{f}\left(\boldsymbol{s}_{t}\right)+\boldsymbol{g}\left(\boldsymbol{s}_{t}\right)\boldsymbol{u}_{t}$, a necessary condition for the existence of (\ref{eq: iter1_min}) is 
\begin{equation}
    \boldsymbol{u}^{1} = - \frac{1}{2}\boldsymbol{H}^{-1}\boldsymbol{g}\left(\boldsymbol{s}_{t}\right)^T\frac{\partial \mathbb{V}^{0}\left(\boldsymbol{s_{t+1}}\right)}{\partial \boldsymbol{s_{t+1}}} \label{eq: u1_min}
\end{equation}
Substituting (\ref{eq: u1_min}) back into (\ref{eq: iter1_min} yields
\begin{eqnarray}
    \mathbb{V}^{0}\left(\boldsymbol{s_{t+1}}\right)-\mathbb{V}^{0}\left(\boldsymbol{s_{t}}\right) &=&  -\left(\boldsymbol{x}_{t}-\boldsymbol{x}_{m,t}\right)^T\boldsymbol{G} \left(\boldsymbol{x}_{t}-\boldsymbol{x}_{m,t}\right) \nonumber\\
    && - \frac{1}{4}\left(\frac{\partial \mathbb{V}^{0}\left(\boldsymbol{s_{t+1}}\right)}{\partial \boldsymbol{s_{t+1}}}\right)^T\boldsymbol{g}\left(\boldsymbol{s}_{t}\right) \nonumber\\
    &&\times \boldsymbol{H}^{-1}\boldsymbol{g}\left(\boldsymbol{s}_{t}\right)^T\frac{\partial \mathbb{V}^{0}\left(\boldsymbol{s_{t+1}}\right)}{\partial \boldsymbol{s_{t+1}}}\leq 0 \nonumber
\end{eqnarray}
Hence, $\boldsymbol{u}^{1}$ is a control law which can stabilize the same ASV system (\ref{eq:ASV_Dyn2}), if Assumption \ref{assump:BaselineC} holds. Applying Lemma \ref{lem:Pi_Eval} to $\boldsymbol{u}^{1}$, we can get a new Lyapunov function $\mathbb{V}^{1}\left(\boldsymbol{s_{t}}\right)$. In terms of  $\mathbb{V}^{1}\left(\boldsymbol{s_{t}}\right)$, (\ref{eq: iter1_min}) and (\ref{eq: u1_min}), we can show that $\boldsymbol{u}^{2}$ also stabilizes the ASV system (\ref{eq:ASV_Dyn2}). Repeating (\ref{eq: iter1_min}) and (\ref{eq: u1_min}) for all $i=1$, $2$, $\ldots$, we can prove that all $\boldsymbol{u}^{i}$ can stabilize the ASV system (\ref{eq:ASV_Dyn2}), if Assumption \ref{assump:BaselineC} holds.   
\end{proof}

\section{Simulation} \label{sec:Sim&Exp}
In this section, the proposed learning-based control algorithm is implemented to the trajectory tracking control of a supply ship model presented in \cite{Skjetnea2005Auto, Peng2013}. Model parameters are summarized in Table \ref{tab:HydroTab}. The unmodeled dynamics in the simulations are given by $g_1 = 0.279uv^2 + 0.342 v^2r$, $g_2 = 0.912u^2v$, and $ g_3 = 0.156ur^2+0.278urv^3$, respectively.
The based-line control law $\boldsymbol{u}_b$ is designed based on a nominal model with the following simplified linear dynamics in terms of the backstepping control method \cite{Khalil2002Book, Zhang2018TIE}.
\begin{equation}
    \boldsymbol{M}_m\dot{\boldsymbol{\nu}}_m = \boldsymbol{\tau}-\boldsymbol{D}_m{\boldsymbol{\nu}}_m \label{eq:LinearModel}
\end{equation}
where $\boldsymbol{M}_m=diag\left\{M_{11},\;M_{22},\;M_{33}\right\}$. $\boldsymbol{D}_m=diag\left\{-X_{v},-Y_{v},\;-N_{r}\right\}$.  The reference signal is assumed to be produced by the following motion planner.
\begin{equation}
       \dot{\boldsymbol{\eta}}_r = \boldsymbol{R}\left(\boldsymbol{\eta}_r\right)\boldsymbol{\nu}_r \quad
       \dot{\boldsymbol{\nu}}_r= \boldsymbol{a}_r \label{eq:MotionPlanner}
\end{equation}
where ${\boldsymbol{\eta}}_r=\left[x_r,y_r,\psi_r\right]^T$ is the generalized reference position vector, ${\boldsymbol{\nu}}_r=\left[u_r,0,r_r\right]^T$ is the generalized reference velocity vector, and $\boldsymbol{a}_r=\left[\dot{u}_r,0,\dot{r}_r\right]^T$. In the simulation, the initial position vector  ${\boldsymbol{\eta}}_r\left(0\right)$ is chosen to be ${\boldsymbol{\eta}}_r\left(0\right)=\left[0,0,\frac{\pi}{4}\right]^T$, and we set $u_r\left(0\right)=0.4$ $m/s$ and $r_r\left(0\right)=0$ $rad/s$. The reference acceleration $\dot{u}_r$ and angular rates are  chosen to be
\begin{eqnarray}
    \dot{u}_r&=&\left\{\begin{array}{ll}
    0.005\; m/s^2 &\text{if }  t<20\; s \\
    0    \; m/s^2 & \text{otherwise}
    \end{array}
    \right. \\
    \dot{r}_r&=&\left\{\begin{array}{ll}
    \frac{\pi}{600}\; rad/s^2 \; & \;  \text{if } 25\; s \leq t<50 \;s \\
    0    \; rad/s^2  & \text{otherwise}
    \end{array}
    \right.
\end{eqnarray}

\begin{table}[tbph]
    \centering
    \caption{Model parameters} \vspace{-2mm} \label{tab:HydroTab}
    \begin{tabular}{cc|cc}
    \toprule
       Parameters  & Values & Parameters  & Values \\ \toprule
        $m$ & $23.8$ & $Y_{\dot{r}}$ & $-0.0$ \\
        $I_z$ & $1.76$ & $Y_{{r}}$ & $0.1079$ \\
        $x_g$ & $0.046$ & $Y_{\vert{v}\vert {r}}$ & $-0.845$ \\
        $X_{\dot{u}}$ & $-2.0$ & $Y_{\vert{r}\vert {r}}$ & $-3.45$ \\
        $X_{{u}}$ & $-0.7225$ & $N_{{v}}$ & $-0.1052$ \\
        $X_{\vert{u}\vert {u}}$ & $-1.3274$ & $N_{\vert{v}\vert {v}}$ & $5.0437$ \\
        $X_{{u}{u}{u}}$ & $-1.8664$  & $N_{\vert{r}\vert {v}}$ & $-0.13$ \\
        $Y_{\dot{v}}$ & $-10.0$ & $N_{\dot{r}}$ & $-1.0$ \\
        $Y_{{v}}$ & $-0.8612$ & $N_{{r}}$ & $-1.9$ \\
        $Y_{\vert{v}\vert {v}}$ & $-36.2823$ & $N_{\vert{v}\vert {r}}$ & $0.08$ \\
        $Y_{\vert{r}\vert {v}}$ & $-0.805$ & $N_{\vert{r}\vert {r}}$ & $-0.75$  \\ \toprule
    \end{tabular}
\end{table}

\begin{table}[tbph]
    \centering
    \caption{Reinforcement learning configurations} \vspace{-2mm} \label{tab:RLTab}
    \begin{tabular}{cc}
    \toprule
      Parameters  & Values  \\ \toprule
        Learning rate $\iota_Q$ &  $0.001$ \\
        Learning rate   $\iota_\pi$ &  $0.0001$ \\
        Learning rate $\iota_\alpha$ &  $0.0001$ \\
        $\kappa$ & $0.01$ \\
        actor neural network & fully connected with two hidden layers \\
        &  (128 neurons per hidden layer) \\
        critic neural networks & fully connected with two hidden layers  \\
        &  (128 neurons per hidden layer) \\
        Replay memory capacity & $1\times 10^{6}$ \\
        Sample batch size &$128$ \\
        $\gamma$ & $0.998$ \\
        Training episodes & $1001$ \\
        Steps per episode & $1000$  \\
        time step size $\delta t$ & $0.1$  \\\toprule
    \end{tabular}
\end{table}
At the training stage, we uniformly randomly sample $x\left(0\right)$ and $y\left(0\right)$ from $\left(-1.5, 1.5\right)$,  $\psi\left(0\right)$ from $\left(0.1\pi, 0.4\pi\right)$ and $u\left(0\right)$ from $\left(0.2, 0.4\right)$, and we choose $v\left(0\right)=0$ and $r\left(0\right)= 0$. The proposed control algorithm is compared with two benchmark designs: the baseline control $\boldsymbol{u}_0$ and the RL control without $\boldsymbol{u}_0$. Configurations for the training and neural networks are found in Table \ref{tab:RLTab}. The matrix $\boldsymbol{G}$ and $\boldsymbol{H}$ are chosen to be $\boldsymbol{G}=diag\left\{0.025,0.025, 0.0016, 0.005,0.001, 0\right\}$ and $\boldsymbol{H}=diag\left\{1.25e^{-4},1.25e^{-4}, 8.3e^{-5}\right\}$, respectively. 
\begin{figure}[tbp]
    \centering
    \includegraphics[width=0.45\textwidth]{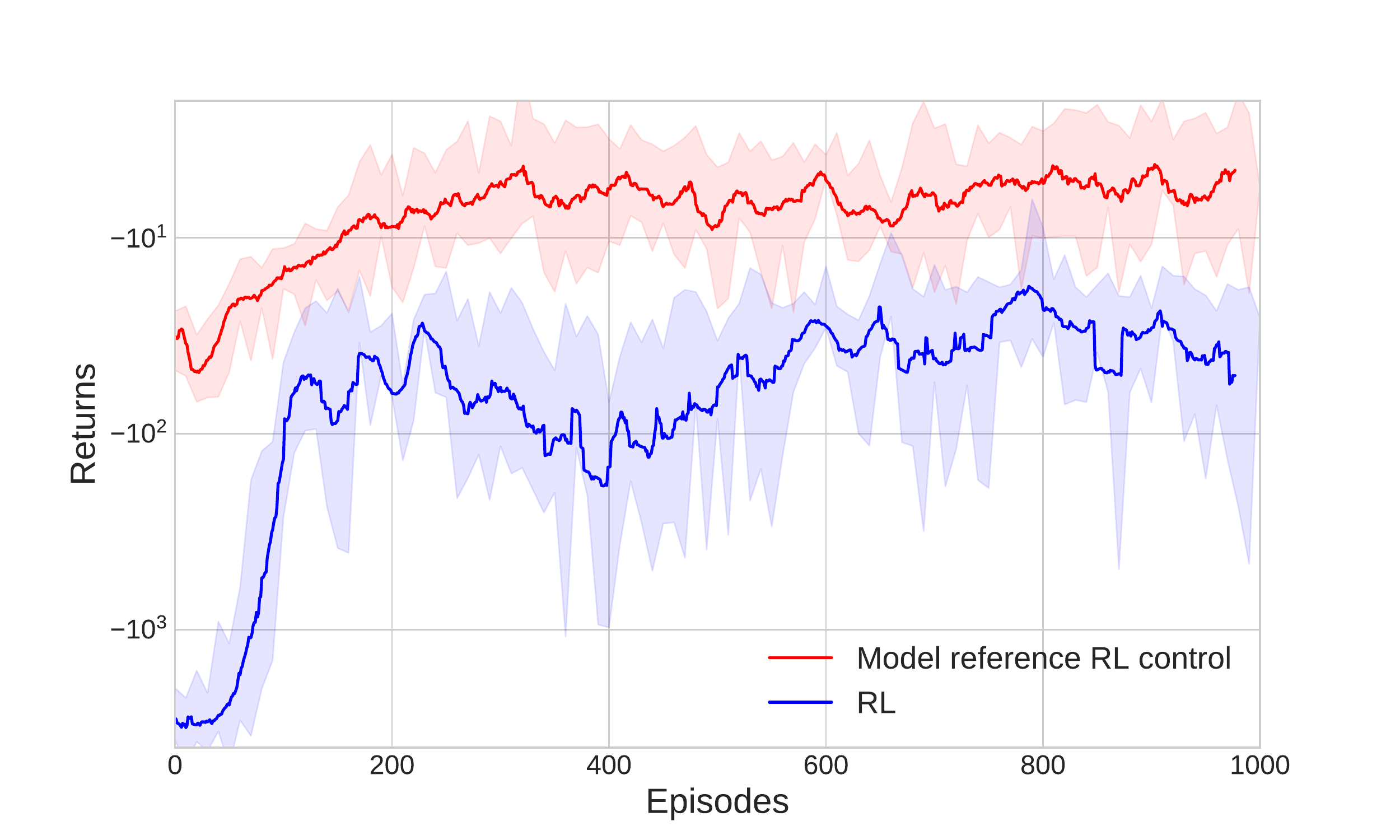}
    \caption{Learning curves of two RL algorithms at training (One episode is a training trial, and $1000$ time steps per episode) }
    \label{fig:learningCurve}
\end{figure}

At the training stage, we run the ASV system for $100$ $s$, and the repeat the training processes for $1000$ times (i.e., $1000$ episodes). 
Figure \ref{fig:learningCurve} shows the learning curves of the proposed algorithm (red) and the RL algorithm without baseline control (blue). The learning curves demonstrate that both of the two algorithms will converge in terms of the long term returns. However, our proposed algorithm results in a larger return (red) in comparison with the RL without baseline control (blue). Hence, the introduction of the baseline control helps to increase the sample efficiency significantly, as the proposed algorithm (blue) converges faster to a higher return value.  

\begin{figure*}[h]
  \subfloat[Model reference reinforcement learning control]{
	\begin{minipage}[c][1\width]{
	   0.32\textwidth}
	   \centering
	   \includegraphics[width=1\textwidth]{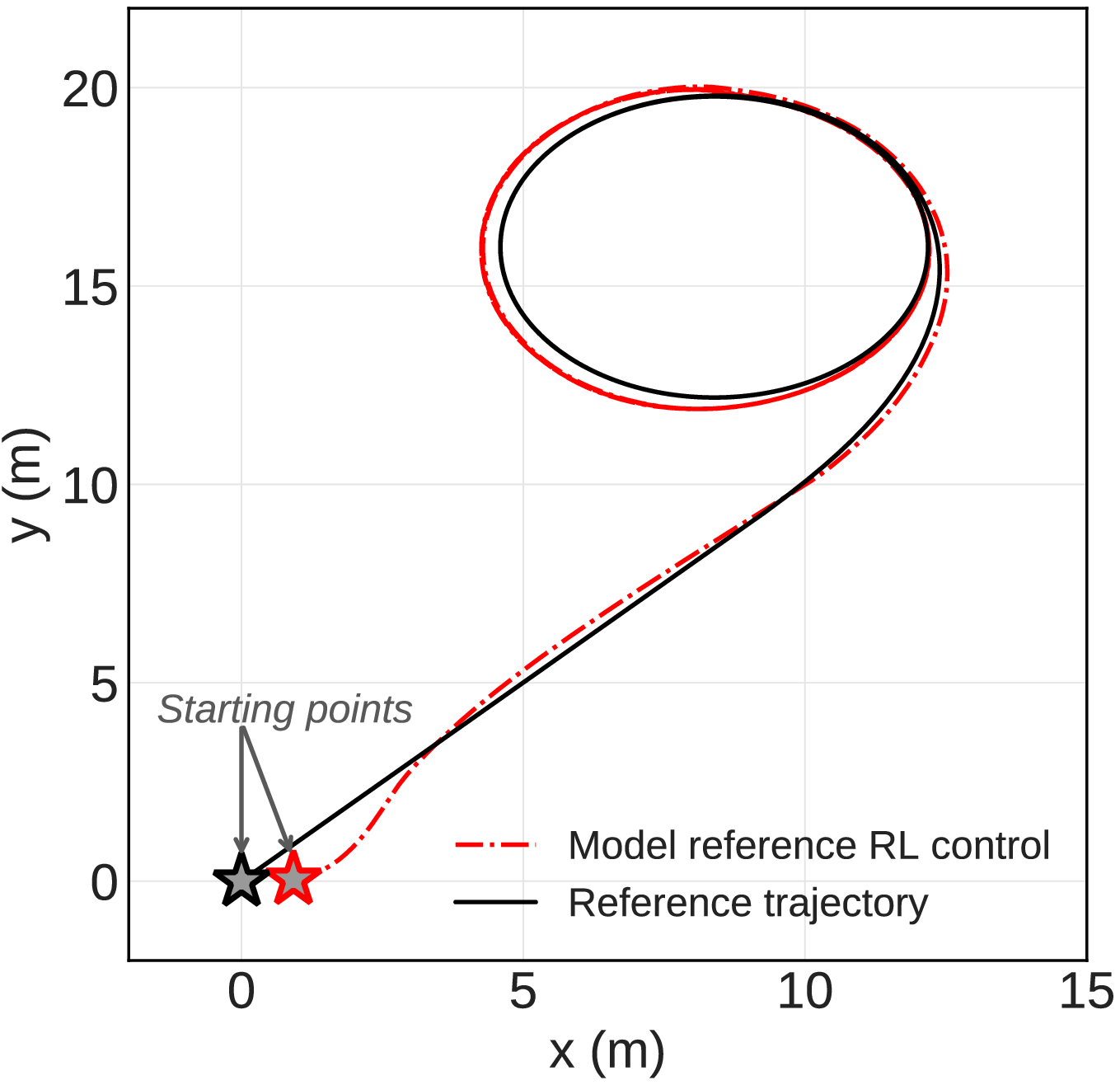}
	\end{minipage}}
 \hfill 	
  \subfloat[Only deep reinforcement learning]{
	\begin{minipage}[c][1\width]{
	   0.32\textwidth}
	   \centering
	   \includegraphics[width=1\textwidth]{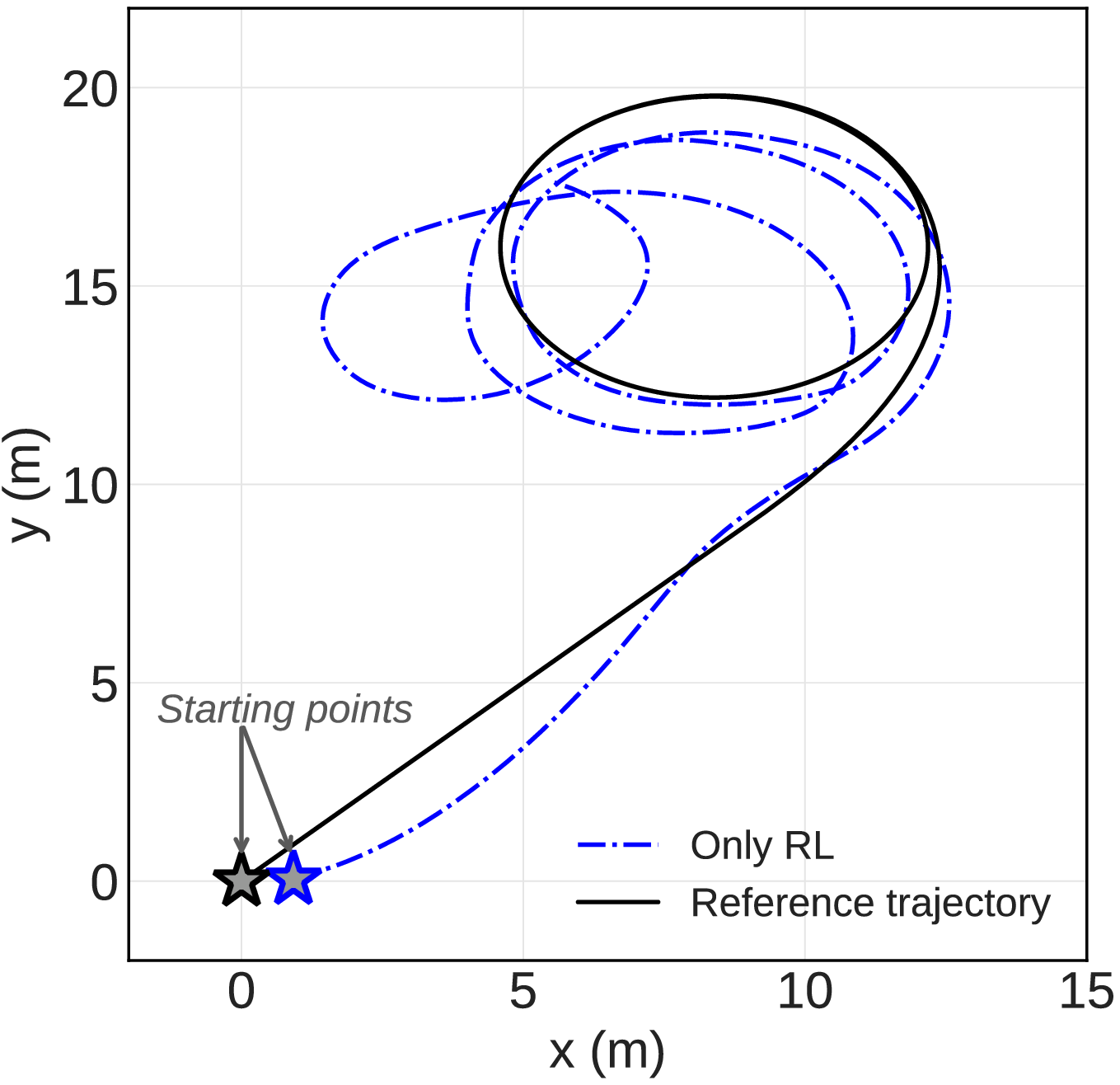}
	\end{minipage}}
 \hfill	
  \subfloat[Only baseline control]{
	\begin{minipage}[c][1\width]{
	   0.32\textwidth}
	   \centering
	   \includegraphics[width=1\textwidth]{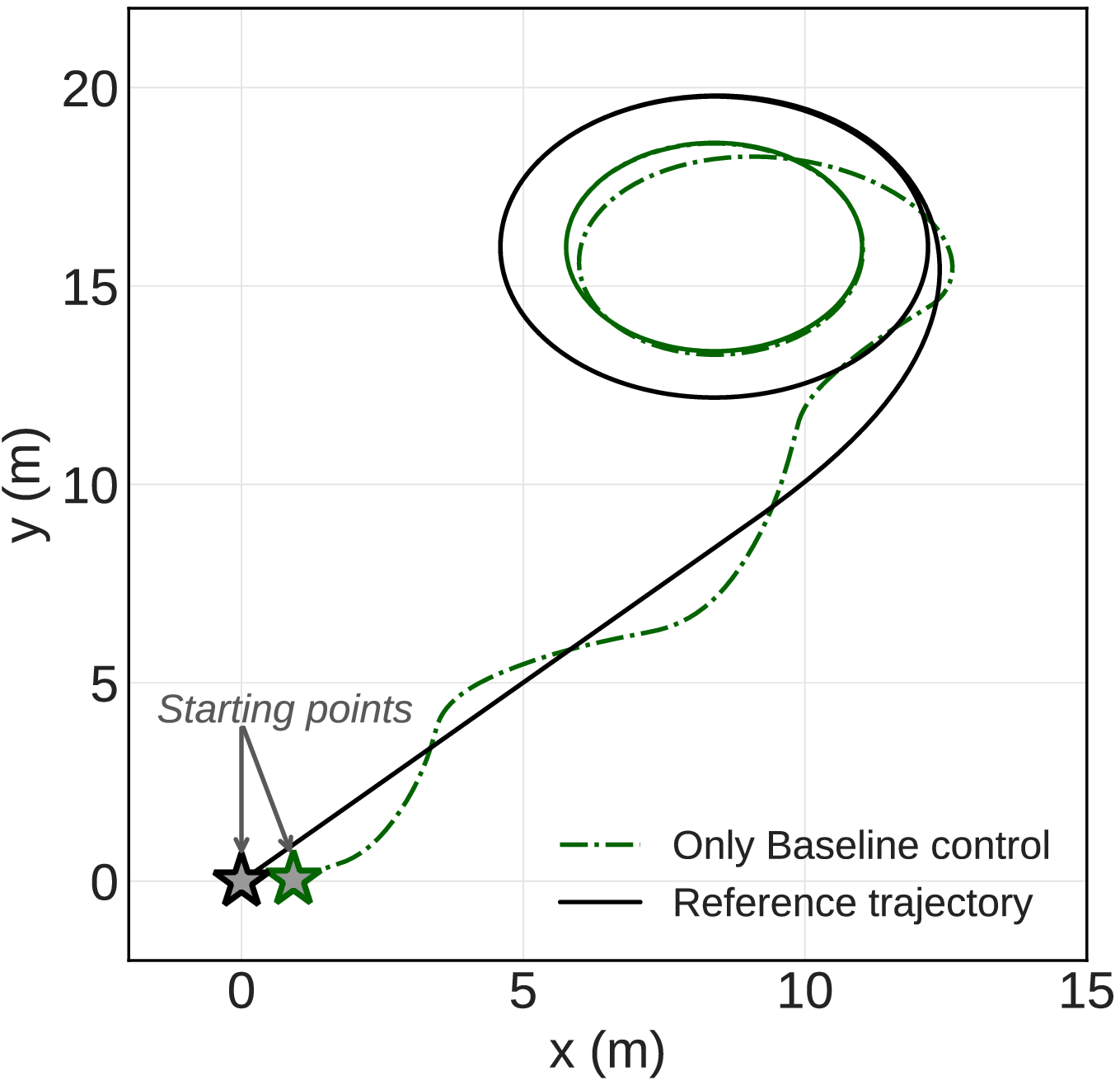}
	\end{minipage}}
\caption{Trajectory tracking results of the three algorithms  (The first evaluation)}
\label{fig:TrajTracking}
\end{figure*}

At the first evaluation stage, we run the ASV system for $200$ $s$ to demonstrate whether the control law can ensure stable trajectory tracking. Note that we run the ASV for $100$ $s$ at training.  The trajectory tracking performance of the three algorithms (our proposed algorithm, the baseline control $\boldsymbol{u}_0$, and only RL control) is shown in Figures \ref{fig:TrajTracking}. As observed from Figure \ref{fig:TrajTracking}.b, the control law learned merely using deep RL fails to ensure stable tracking performance. It implies that only deep RL cannot ensure the closed-loop stability. In addition, the baseline control itself fails to achieve acceptable tracking performance mainly due to the existence of system uncertainties. By combining the baseline control and deep RL,  the trajectory tracking performance is improved dramatically, and the closed-loop stability is also ensured. The position tracking errors are summarized in Figure \ref{fig:ErrorX} and \ref{fig:ErrorY}. Figure \ref{fig:ErrorDist} shows the absolute distance errors used to compare the tracking accuracy of the three algorithms. The introduction of the deep RL increases the tracking performance substantially.
\begin{figure}
    \centering
    \includegraphics[width=0.45\textwidth]{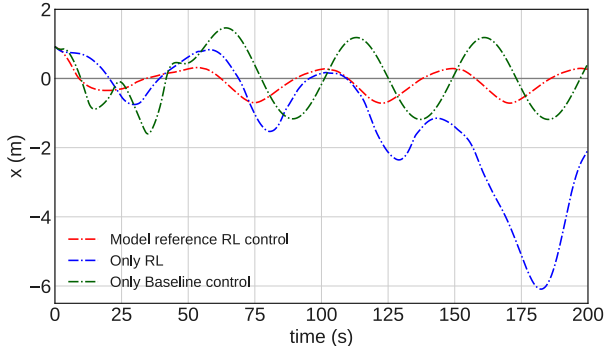}
    \caption{Position tracking errors ($e_x$)}
    \label{fig:ErrorX}
\end{figure}
\begin{figure}
    \centering
    \includegraphics[width=0.45\textwidth]{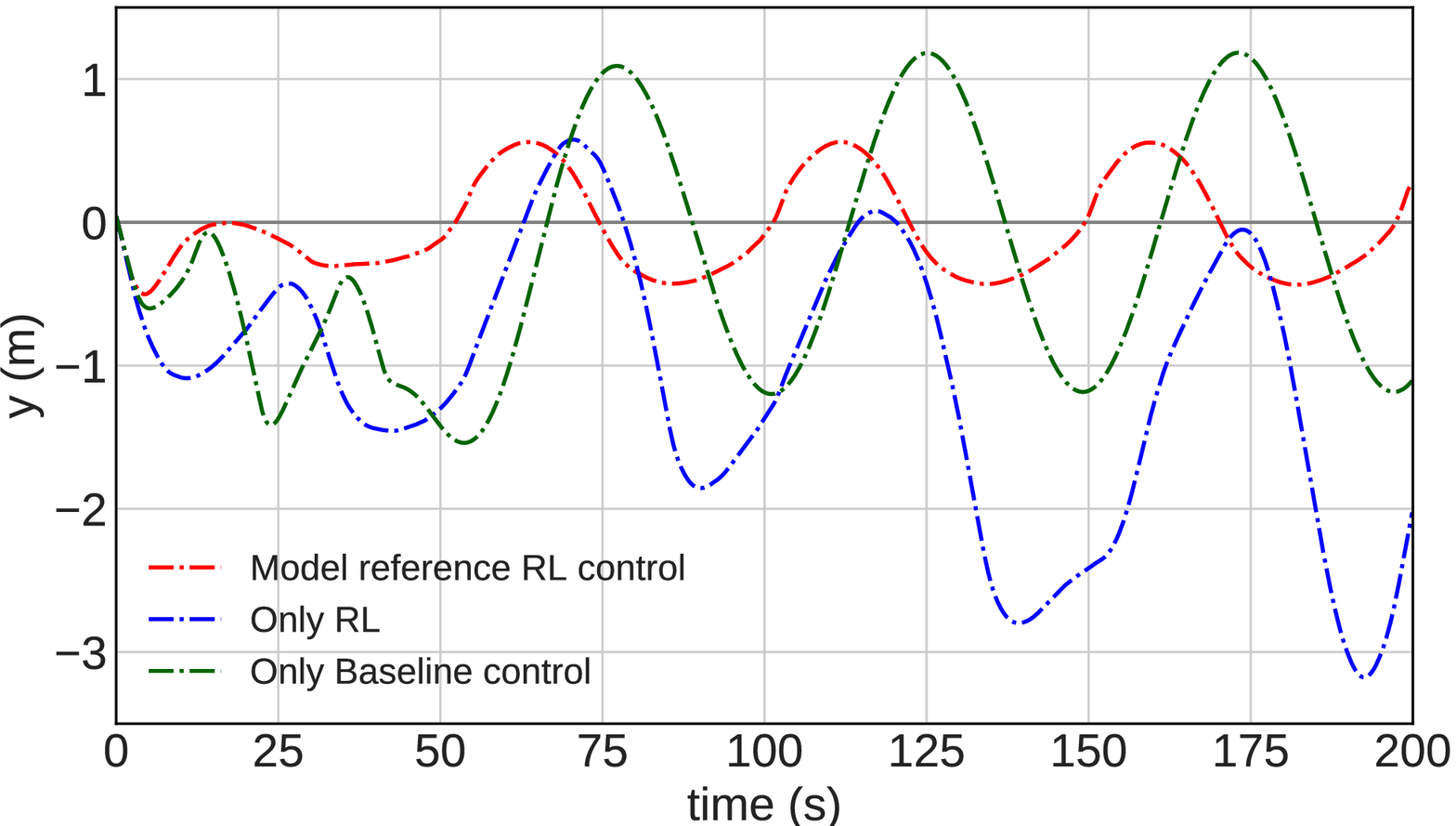}
    \caption{Position tracking errors ($e_y$)}
    \label{fig:ErrorY}
\end{figure}
\begin{figure}
    \centering
    \includegraphics[width=0.45\textwidth]{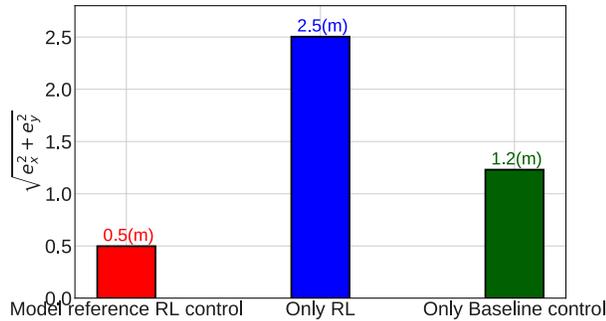}
    \caption{Mean absolute distance errors ($\sqrt{e_x^2+e_y^2}$)}
    \label{fig:ErrorDist} 
\end{figure}

At the second evaluation, we still run the ASV system for $200$ $s$, but change the reference trajectory. Note that we use the same learned control laws in both the first and the second evaluations. In the second evaluation, the reference angular acceleration is changed to
\begin{equation}
    \dot{r}_r=\left\{\begin{array}{ll}
    \frac{\pi}{600}\; rad/s^2 &\text{if }\;\; 25\;s \leq  t<50\;s \\
    -\frac{\pi}{600}\; rad/s^2 &\text{if } 125 \;s\leq  t<150\;s \\
    0    \; rad/s^2  & \text{otherwise}
    \end{array}
    \right.
\end{equation}
The trajectory tracking results are illustrated in Figure \ref{fig:TrajTracking_Case2}. Apparently, the proposed control algorithm can ensure closed-loop stability, while the vanilla RL fails to do so. A better tracking performance is obtained by the proposed control law in comparison with only baseline control.
\begin{figure*}[h]
  \subfloat[Model reference reinforcement learning control]{
	\begin{minipage}[c][1\width]{
	   0.32\textwidth}
	   \centering
	   \includegraphics[width=1\textwidth]{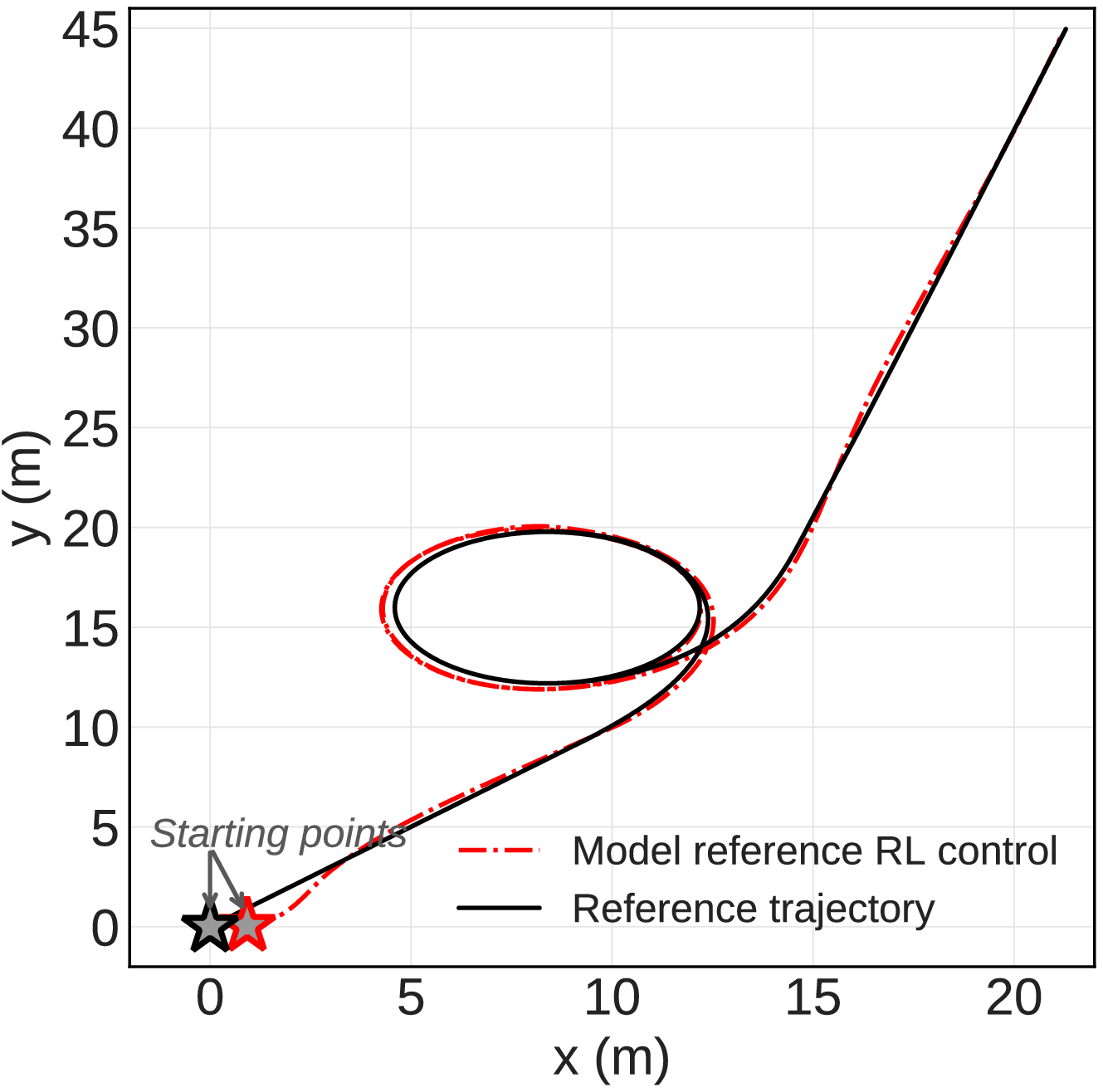}
	\end{minipage}}
 \hfill 	
  \subfloat[Only deep reinforcement learning]{
	\begin{minipage}[c][1\width]{
	   0.32\textwidth}
	   \centering
	   \includegraphics[width=1\textwidth]{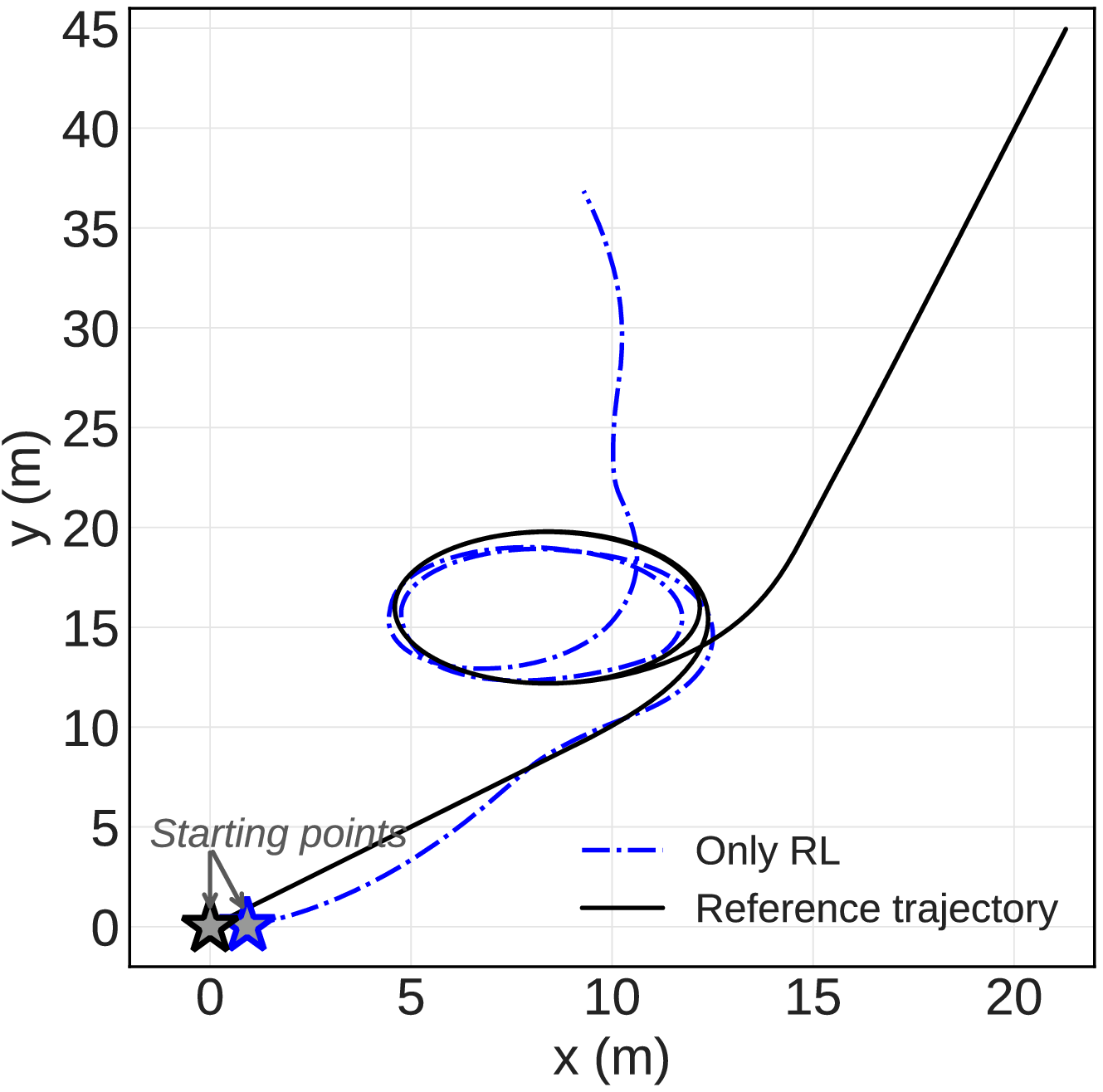}
	\end{minipage}}
 \hfill	
  \subfloat[Only baseline control]{
	\begin{minipage}[c][1\width]{
	   0.32\textwidth}
	   \centering
	   \includegraphics[width=1\textwidth]{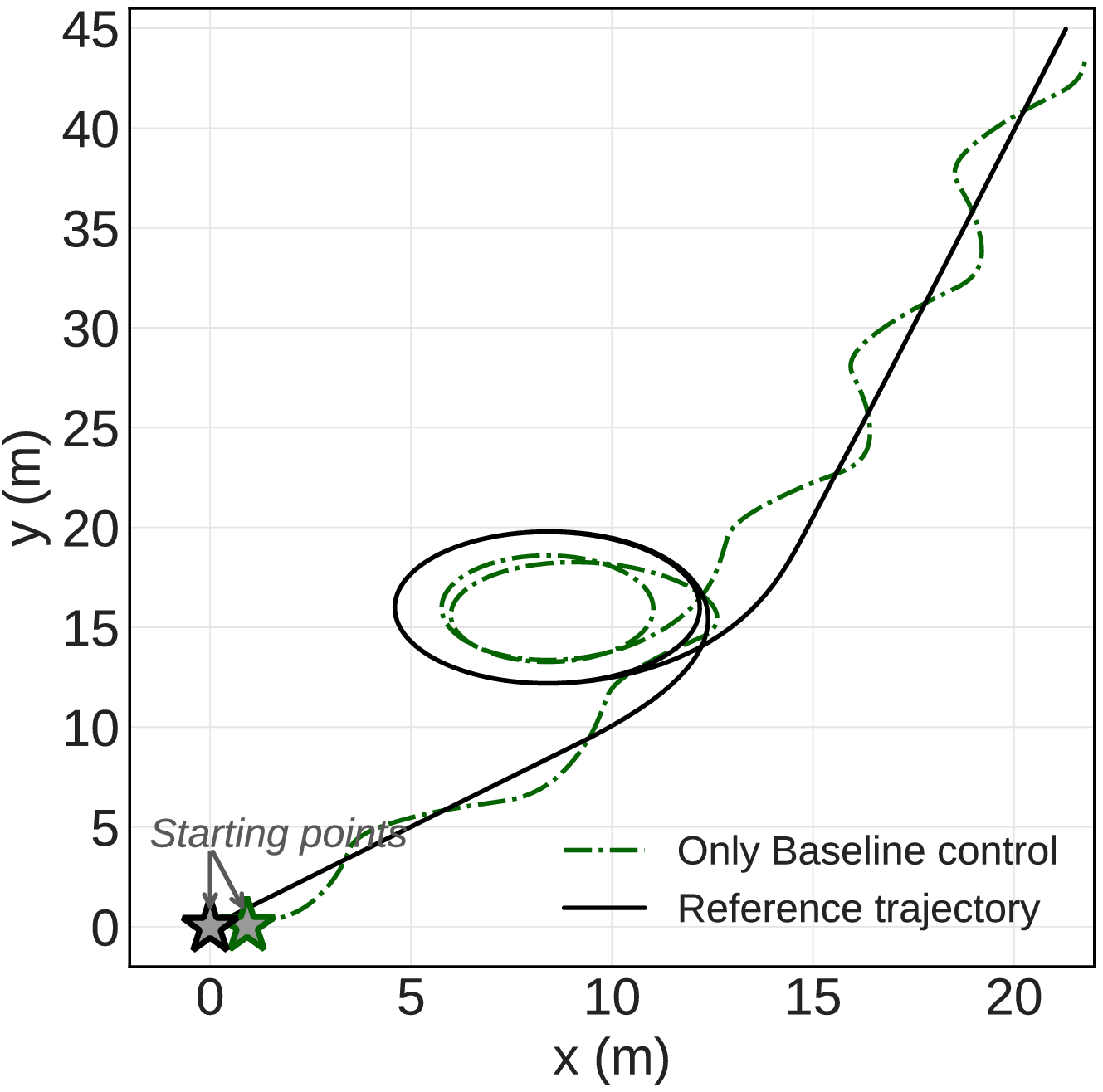}
	\end{minipage}}
\caption{Trajectory tracking results of the three algorithms (The second evaluation)}
\label{fig:TrajTracking_Case2}
\end{figure*}

\section{Conclusions} \label{sec:Concl}
In this paper, we presented a novel learning-based algorithm for the control of uncertain ASV systems by combining a conventional control method with deep reinforcement learning. With the conventional control, we ensured the overall closed-loop stability of the learning-based control and increase the sample efficiency of the deep RL. With the deep RL, we learned to compensate for the model uncertainties, and thus increased the trajectory tracking performance. In the future works, we will extend the results with the consideration of environmental disturbances. The theoretical results will be further verified via experiments instead of simulations. Sample efficiency of the proposed algorithm will also be analyzed.

\bibliography{References}
\bibliographystyle{IEEEtran}

\end{document}